\documentclass[3p,times]{elsarticle}

\usepackage{ecrc}

\volume{}

\firstpage{1}

\journalname{}

\runauth{K. Phalakarn et al.}

\jid{procs}

\jnltitlelogo{}

\usepackage{amssymb}
\usepackage{amsthm}

\usepackage[figuresright]{rotating}

\usepackage{amsmath,enumitem,relsize,tikz}
\usepackage[ruled,linesnumbered,noend]{algorithm2e}
\usetikzlibrary{automata,arrows.meta,shapes.geometric,calc}
\tikzset{->,>=Stealth,every state/.style={thick}}

\newtheorem{theorem}{Theorem}
\newtheorem{lemma}[theorem]{Lemma}
\newtheorem{proposition}[theorem]{Proposition}
\theoremstyle{definition}
\newtheorem*{example*}{Example}
\newtheorem{definition}[theorem]{Definition}

\newcommand\fakeslant[1]{\pdfliteral{1 0 0.3 1 0 0 cm}\pdfliteral{1 0 0 1 0 -0.2 cm}#1\pdfliteral{1 0 0 1 0 0.2 cm}\pdfliteral{1 0 -0.3 1 0 0 cm}}
\newcommand\mathbbsl[1]{\mathbb{\fakeslant{#1}}}
\newcommand{\LL}{{\mathbbsl{L}}}
\newcommand{\AP}{{\mathit{AP}}}
\newcommand{\LTL}[1]{\textsf{\upshape #1}}
\DeclareMathOperator{\XX}{\LTL{X}}
\DeclareMathOperator{\until}{\LTL{U}}
\DeclareMathOperator{\always}{\LTL{G}}
\DeclareMathOperator{\eventually}{\LTL{F}}
\DeclareMathOperator{\buchi}{\LTL{GF}}
\DeclareMathOperator{\cobuchi}{\LTL{FG}}

\newcommand{\change}[1]{{{#1}}}

\begin{document}

\begin{frontmatter}




\title{Strategy Templates for Almost-Sure and Positive Winning of Stochastic Parity Games towards Permissive and Resilient Control\tnoteref{tnt}}
\tnotetext[tnt]{This work is an extended version of~\cite{DBLP:conf/ictac/PhalakarnPH24}. K.~Phalakarn and I.~Hasuo are supported by the ERATO HASUO Metamathematics for Systems Design Project grant No.~JPMJER1603 and the ASPIRE grant No.~JPMJAP2301, JST. S.~Pruekprasert is supported by the KAKENHI grant No.~JP22KK0155, JSPS. This work was partly done while S.~Pruekprasert was affiliated with the National Institute of Advanced Industrial Science and Technology, Tokyo, Japan.}


\author[nii]{Kittiphon Phalakarn\corref{cor}}\ead{kphalakarn@nii.ac.jp}
\author[utokyo]{Sasinee Pruekprasert}\ead{spruekprasert@g.ecc.u-tokyo.ac.jp}
\author[nii,sokendai,imiron]{Ichiro Hasuo}\ead{hasuo@nii.ac.jp}
\cortext[cor]{Corresponding author}

\address[nii]{National Institute of Informatics, Tokyo, Japan}
\address[utokyo]{The University of Tokyo, Tokyo, Japan}
\address[sokendai]{SOKENDAI (The Graduate University for Advanced Studies), Kanagawa, Japan}
\address[imiron]{Imiron Co., Ltd., Tokyo, Japan}

\begin{abstract}
\emph{Stochastic games} are fundamental in various applications, including the control of cyber-physical systems (CPS), where both controller and environment are modeled as players. Traditional algorithms typically aim to determine a single \emph{winning strategy} to develop a controller. However, in CPS control and other domains, permissive controllers are essential, as they enable the system to adapt when additional constraints arise and remain resilient to runtime changes. This work generalizes the concept of \emph{(permissive winning) strategy templates}, originally introduced by Anand et al.~at TACAS and CAV 2023 for deterministic games, to incorporate stochastic games. These templates capture an infinite number of winning strategies, allowing for efficient strategy adaptation to system changes. We focus on two winning criteria (almost-sure and positive winning) and five winning objectives (safety, reachability, Büchi, co-Büchi, and parity). Our contributions include algorithms for constructing templates for each winning criterion and objective and a novel approach for extracting a winning strategy from a given template. Discussions on comparisons between templates and between strategy extraction methods are provided.
\end{abstract}

\begin{keyword}
stochastic game \sep parity game \sep strategy template \sep game-based control \sep permissive controller \sep resiliency
\end{keyword}

\end{frontmatter}


\section{Introduction}

Games on graphs are fundamental in control theory and cyber-physical system (CPS) design~\cite{DBLP:journals/access/TusharYSNASP23}, providing a robust framework for analyzing and developing systems that dynamically interact with their environments. In this context, game-based controllers apply game theory principles to effectively manage interactions between systems and their environments. Control problems are often represented as two-player games between the controller and the environment\change{~\cite{DBLP:journals/access/TusharYSNASP23, DBLP:journals/arcras/MardenS18, DBLP:journals/automatica/LvXJLY24, DBLP:journals/tac/PruekprasertUK16}}: the controller player aims at influencing the system's behavior toward its specified objectives, and the environment player introduces uncertainties and external factors that challenge the controller's decision-making.

\emph{Stochastic games} extend the conventional two-player game framework by introducing probabilistic transitions to capture uncertainty in system dynamics. Conceptually referred to as ``2.5''-player games, stochastic games involve two primary players along with an additional ``0.5'' player, representing the random or stochastic behavior of the environment. Under this setting, players must develop strategies that consider both their opponents' actions and the probabilistic transitions. \change{While stochastic games can capture only limited aspects of uncertainty---due to finite random choices---they nonetheless provide an expressive framework for modeling probabilistic and nondeterministic behaviors in control systems. Consequently, they are widely applied in various domains\change{~\cite{DBLP:journals/access/TusharYSNASP23, DBLP:journals/arcras/MardenS18, DBLP:journals/ejcon/SvorenovaK16}}, particularly in system control and theoretical computer science, where they are used to analyze probabilistic systems and programs.}

\subsection{Related Works}
Traditional game-solving algorithms typically aim to identify a single winning strategy for each player, without explicitly accounting for the strategy's \emph{permissiveness}. However, permissive controllers play a vital role in real-world applications. The concept of permissiveness in control theory, particularly in supervisory control, was formally introduced by Ramadge and Wonham in 1987~\cite{ramadge1987supervisory} and is often regarded as the classical definition of permissiveness. According to their work, a controller is considered more permissive (or less restrictive) than the other if it allows all behaviors permitted by the latter without disabling any additional system behaviors. This notion enhances system flexibility, enabling system adaptation to additional constraints and unpredictable operational condition changes during runtime.

The classical notion of permissiveness has inspired the development of various related concepts in permissive control. For instance, some approaches penalize the controller based on the disable costs associated with each control action~\cite{sengupta1998optimal,DBLP:journals/automatica/LvXJLY24}, while some others focus on maximally permissive controllers constrained by the number of allowable losing loops for the controller player~\cite{DBLP:journals/tac/PruekprasertUK16}. Additionally, permissiveness plays a fundamental role in resilient control and helps addressing uncertainties in a wide range of applications. These include flexible manufacturing systems~\cite{DBLP:journals/automatica/ChenL11,DBLP:journals/ijaac/RezigGAR19}, warehouse automation~\cite{tatsumoto2018application}, and resilient control against potential attacks in CPS~\cite{DBLP:journals/csysl/MaC22}.

In the context of games, the classical concept of permissiveness has been explored in various settings. In parity games, Bernet et al.~\cite{DBLP:journals/ita/BernetJW02} demonstrated that a maximally permissive strategy exists when considering only memoryless strategies and provided an algorithm to construct such a strategy. Another perspective on permissiveness was studied in Muller games~\cite{DBLP:journals/tcs/NeiderR014}, where a maximally permissive strategy is constrained to winning strategies that allow visiting losing loops at most twice. Quantitative measures of permissiveness were introduced in~\cite{DBLP:conf/concur/BouyerDMR09,DBLP:conf/atva/BouyerMOU11}, quantifying permissiveness based on the weight of transitions disabled by strategies. Moreover, the concept of \emph{weakest} strategies was considered in safety games with imperfect information~\cite{DBLP:conf/tacas/KuijperP09}, leading to a compositional control synthesis method for the weakest safety controllers under partial observation~\cite{DBLP:conf/concur/KuijperP09}. A related notion, termed the \emph{most general} strategy, was presented for compositional controller construction~\cite{DBLP:journals/acta/0001BK15}. The work proposed \emph{decision function templates}, which define all legal control choices for a given observable history. A most general controller is then synthesized using these templates along with a suitable fairness condition to ensure that legal choices are fairly chosen.

Recently, Anand et al.~established a novel concept of strategy templates that offer greater compositionality compared to previous approaches~\cite{DBLP:conf/tacas/AnandMNS23,DBLP:conf/cav/AnandNS23}. In the first work~\cite{DBLP:conf/tacas/AnandMNS23}, they introduced \emph{adequately permissive assumptions} on the opposing player, representing other distributed components. These assumptions, expressed as linear temporal logic (LTL) formulae over vertices and edges, are considered \emph{adequately permissive} if they allow all feasible cooperative system behaviors necessary to achieve the desired objective. This idea was later developed into \emph{permissive winning strategy templates} for deterministic zero-sum games in~\cite{DBLP:conf/cav/AnandNS23}. Their experimental results highlight two key applications of these strategy templates. First, when additional objectives arrive after a winning strategy has been computed, the strategy templates enable faster adaptation of the strategy compared to recomputing it from scratch. Second, the strategy templates facilitate fault-tolerant control by producing new strategies when certain actions become unavailable due to system faults at runtime. In essence, strategy templates effectively account for both evolving requirements and system changes.

The applications of strategy templates extend beyond permissive and resilient control. Recent research has broadened their scope in several directions. One approach involves leveraging strategy templates to solve infinite-state games~\cite{DBLP:conf/cav/SchmuckHDN24}. Another utilizes them as abstractions for synthesizing low-level continuous-time dynamical systems~\cite{DBLP:journals/OJCSYS/NayakERSJ23}. Additionally, studies such as~\cite{DBLP:conf/hybrid/AnandSN24,DBLP:conf/tacas/NayakS24} have explored their role in co-synthesis for multi-player games, where templates function as contracts or constraints between players. These findings highlight the practicality and versatility of strategy templates. For a comprehensive overview of their applications, see~\cite{DBLP:conf/atva/AnandNS24}.

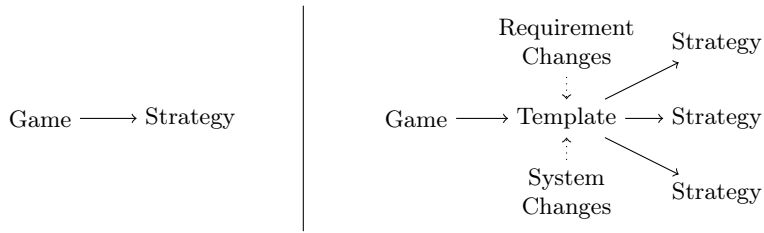
\begin{figure}[t]
\centering
\begin{tikzpicture}
    \node (g1) at (-1,0) {Game};
    \node (s1) at (2,0) {Strategy};
    \node (g2) at (5,0) {Game};
    \node (t2) at (8,0) {Template};
    \node (s21) at (11,1) {Strategy 1};
    \node (s22) at (11,0) {Strategy 2};
    \node (s23) at (11,-1) {Strategy 3};
    \node[align=center] (c1) at (5,1) {Requirement\\Changes};
    \node[align=center] (c2) at (5,-1) {System\\Changes};
    \draw[->] (g1) -- node[align=center,above]{\footnotesize construct} (s1);
    \draw[->] (g2) -- node[align=center,above]{\footnotesize construct} (t2);
    \draw[->] (t2) -- node[above,sloped]{\footnotesize extract} (s21);
    \draw[->] (t2) -- node[above]{\footnotesize extract} (s22);
    \draw[->] (t2) -- node[above,sloped]{\footnotesize extract} (s23);
    \draw[->, dashed] (c1) -- node[above]{\footnotesize modify} (8,1) -- (t2);
    \draw[->, dashed] (c2) -- node[below]{\footnotesize modify} (8,-1) -- (t2);
    \draw[-] (3.5,1.75) -- (3.5,-1.75);
\end{tikzpicture}
\caption{Left: A conventional winning strategy construction, giving one strategy. Right: An overview of a winning strategy construction utilizing a strategy template, allowing strategy adaptation for requirement and system changes.}\label{concept}
\end{figure}

\subsection{Contributions}
To the best of our knowledge, no permissive winning strategy templates have been proposed for stochastic games. In this work, we expand the concept of strategy templates from prior research, as depicted in Fig.~\ref{concept}, to incorporate stochastic games. Our key contributions are as follows.
\begin{enumerate}
    \item We present algorithms to construct strategy templates for \emph{almost-sure winning} criterion under five winning objectives of stochastic games (safety, reachability, Büchi, co-Büchi, and parity). \change{These objectives---often formulated in LTL---are widely studied in the formal methods and system control communities, as they capture fundamental classes of temporal system behaviors~\cite{DBLP:journals/arcras/BeltaS19,DBLP:books/daglib/0020348}.} To achieve these, we incorporate set operations of~\cite{DBLP:conf/tacas/BanerjeeMMSS22} and gadgets of~\cite{DBLP:conf/csl/ChatterjeeJH03}. The correctness proofs are provided. (Sect.~\ref{sec3})
    \item We develop an algorithm to construct strategy templates for \emph{positive winning} criterion under five winning objectives of stochastic games. As more information is required to compose positive winning strategy templates, we also introduce an algorithm to compose such templates. (Sect.~\ref{sec99})
    \item We discuss comparisons between templates based on their \emph{permissiveness} and \emph{sizes}. (Sect.~\ref{sec6})
    \item We propose a novel procedure to extract strategies from strategy templates which balances between the winning objective and the permissiveness. (Sect.~\ref{sec5})
\end{enumerate}
In addition, we redefine the concept of strategy templates and their permissiveness using sets of edges, LTL formulae, and formal languages. (Sect.~\ref{sec3} and~\ref{sec6})

\change{Preliminary results on almost-sure winning strategy templates appeared in an earlier version of this work~\cite{DBLP:conf/ictac/PhalakarnPH24}. This extended version presents the main novel contributions: the development of positive winning strategy templates and corresponding template composition algorithms. In addition, we offer further intuitions and detailed explanations for the algorithms used and proposed throughout the work.}

\section{Preliminaries}

\subsection{Linear Temporal Logic}
We briefly review \emph{linear temporal logic (LTL)}, which is later used to define winning objectives. We invite interested readers to see~\cite{DBLP:books/daglib/0020348} for a formal definition.

\begin{definition}[Linear Temporal Logic Formula]\label{defn:LTL}
\emph{LTL formulae} over the set $\AP$ of atomic propositions are formed by the following grammar, where $a \in \AP$.
$$\varphi ::= \textnormal{true} \mid a \mid \varphi_1 \wedge \varphi_2 \mid \neg\varphi \mid \XX \varphi \mid \varphi_1 \until \varphi_2$$
\end{definition}
\noindent The semantics of LTL over an infinite sequence $\bar{v} = v_0 v_1 \ldots \in \AP^{\,\omega}$ of atomic propositions is defined as follows.
\begin{alignat*}{2}
    \bar{v} &\vDash \textnormal{true}.\\
    \bar{v} &\vDash a
            &&\textnormal{\quad iff } v_0 = a, \textnormal{ for } a \in \AP.
    \\
    \bar{v} &\vDash \varphi_1 \wedge \varphi_2
            &&\textnormal{\quad iff } \bar{v} \vDash \varphi_1 \textnormal{ and } \bar{v} \vDash \varphi_2.
    \\
    \bar{v} &\vDash \neg \varphi
            &&\textnormal{\quad iff } \bar{v} \not\vDash \varphi.
    \\
    \bar{v} &\vDash \XX  \varphi
            &&\textnormal{\quad iff } v_1 v_2 \ldots \vDash \varphi.
    \\
    \bar{v} &\vDash \varphi_1 \until \varphi_2
            &&\textnormal{\quad iff } \exists i \geq 0, v_i v_{i+1} \ldots \vDash \varphi_2 \wedge (\forall j < i, v_j v_{j+1} \ldots \vDash \varphi_1).
\end{alignat*}
We say that $\bar{v}$ \emph{satisfies} an LTL formula $\varphi$ if and only if $\bar{v} \vDash \varphi$. Given $X \subseteq \AP$, we write $\bar{v} \vDash X$ to denote $\bar{v} \vDash \bigvee_{x \in X} x$ for notational convenience. Also, the temporal modalities \emph{eventually} and \emph{always} are defined by $\eventually\varphi := \textnormal{true} \until \varphi$ and $\always\varphi := \neg \eventually(\neg \varphi)$, respectively.

\subsection{Stochastic Games}\label{sec2.2}
The following definition is adapted from~\cite{DBLP:conf/csl/ChatterjeeJH03}.

\begin{definition}[Stochastic Game]
A \emph{stochastic game (SG)} is denoted by $G = (V,E,(V_\square,V_\bigcirc,V_\mathlarger{\triangle}))$ where $(V,E)$ is a finite directed graph and $(V_\square,V_\bigcirc,V_\mathlarger{\triangle})$ is a partition of $V$.
\end{definition}

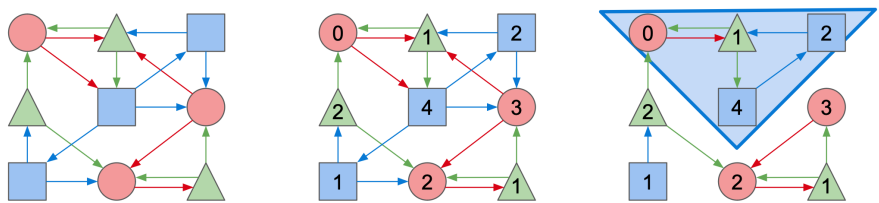
\begin{figure}[t]
\centering
\begin{tikzpicture}[
    cir/.style={anchor=south, minimum size=0pt},
    box/.style={anchor=south, minimum size=25pt, regular polygon, regular polygon sides=4, inner sep=0pt},
    tri/.style={anchor=south, minimum size=23pt, regular polygon, regular polygon sides=3, inner sep=0pt}
]
    \node [state, box] at (  0,  0) (s00) {$\phantom{1}$};
    \node [state, cir] at (1.4,  0) (s01) {$\phantom{2}$};
    \node [state, tri] at (2.8,  0) (s02) {$\phantom{1}$};
    \node [state, tri] at (  0,1.2) (s10) {$\phantom{2}$};
    \node [state, box] at (1.4,1.2) (s11) {$\phantom{4}$};
    \node [state, cir] at (2.8,1.2) (s12) {$\phantom{3}$};
    \node [state, cir] at (  0,2.4) (s20) {$\phantom{0}$};
    \node [state, tri] at (1.4,2.4) (s21) {$\phantom{1}$};
    \node [state, box] at (2.8,2.4) (s22) {$\phantom{2}$};

    \path (s00) edge (s01);
    \path (s00) edge (s10);
    \path ($(s01.0)-(0,0.5ex)$) edge ($(s02.150)-(0.3ex,0.5ex)$);
    \path ($(s02.150)+(0.3ex,0.5ex)$) edge ($(s01.0)+(0,0.5ex)$);
    \path (s02) edge (s12);
    \path (s10) edge (s20);
    \path (s10.330) edge (s01);
    \path (s11.225) edge (s00.45);
    \path (s11.45) edge (s22.225);
    \path (s11) edge (s12);
    \path (s12) edge (s21.330);
    \path (s12) edge (s01);
    \path (s20) edge (s11.135);
    \path ($(s20.0)-(0,0.5ex)$) edge ($(s21.150)-(0.3ex,0.5ex)$);
    \path ($(s21.150)+(0.3ex,0.5ex)$) edge ($(s20.0)+(0,0.5ex)$);
    \path (s21) edge (s11);
    \path (s22) edge (s21.30);
    \path (s22) edge (s12);
\end{tikzpicture}\qquad\qquad\begin{tikzpicture}[
    cir/.style={anchor=south, minimum size=0pt},
    box/.style={anchor=south, minimum size=25pt, regular polygon, regular polygon sides=4, inner sep=0pt},
    tri/.style={anchor=south, minimum size=23pt, regular polygon, regular polygon sides=3, inner sep=0pt}
]
    \node [state, box] at (  0,  0) (s00) {$1$};
    \node [state, cir] at (1.4,  0) (s01) {$2$};
    \node [state, tri] at (2.8,  0) (s02) {$1$};
    \node [state, tri] at (  0,1.2) (s10) {$2$};
    \node [state, box] at (1.4,1.2) (s11) {$4$};
    \node [state, cir] at (2.8,1.2) (s12) {$3$};
    \node [state, cir] at (  0,2.4) (s20) {$0$};
    \node [state, tri] at (1.4,2.4) (s21) {$1$};
    \node [state, box] at (2.8,2.4) (s22) {$2$};

    \path (s00) edge (s01);
    \path (s00) edge (s10);
    \path ($(s01.0)-(0,0.5ex)$) edge ($(s02.150)-(0.3ex,0.5ex)$);
    \path ($(s02.150)+(0.3ex,0.5ex)$) edge ($(s01.0)+(0,0.5ex)$);
    \path (s02) edge (s12);
    \path (s10) edge (s20);
    \path (s10.330) edge (s01);
    \path (s11.225) edge (s00.45);
    \path (s11.45) edge (s22.225);
    \path (s11) edge (s12);
    \path (s12) edge (s21.330);
    \path (s12) edge (s01);
    \path (s20) edge (s11.135);
    \path ($(s20.0)-(0,0.5ex)$) edge ($(s21.150)-(0.3ex,0.5ex)$);
    \path ($(s21.150)+(0.3ex,0.5ex)$) edge ($(s20.0)+(0,0.5ex)$);
    \path (s21) edge (s11);
    \path (s22) edge (s21.30);
    \path (s22) edge (s12);
\end{tikzpicture}\qquad\qquad\begin{tikzpicture}[
    cir/.style={anchor=south, minimum size=0pt},
    box/.style={anchor=south, minimum size=25pt, regular polygon, regular polygon sides=4, inner sep=0pt},
    tri/.style={anchor=south, minimum size=23pt, regular polygon, regular polygon sides=3, inner sep=0pt}
]
    \node [state, box] at (  0,  0) (s00) {$1$};
    \node [state, cir, dotted] at (1.4,  0) (s01) {$2$};
    \node [state, tri, dotted] at (2.8,  0) (s02) {$1$};
    \node [state, tri] at (  0,1.2) (s10) {$2$};
    \node [state, box, ultra thick] at (1.4,1.2) (s11) {$4$};
    \node [state, cir, dotted] at (2.8,1.2) (s12) {$3$};
    \node [state, cir, ultra thick] at (  0,2.4) (s20) {$0$};
    \node [state, tri, ultra thick] at (1.4,2.4) (s21) {$1$};
    \node [state, box, ultra thick] at (2.8,2.4) (s22) {$2$};

    \path (s00) edge (s10);
    \path ($(s01.0)-(0,0.5ex)$) edge ($(s02.150)-(0.3ex,0.5ex)$);
    \path ($(s02.150)+(0.3ex,0.5ex)$) edge ($(s01.0)+(0,0.5ex)$);
    \path (s02) edge (s12);
    \path (s10) edge (s20);
    \path (s10.330) edge (s01);
    \path (s11.45) edge (s22.225);
    \path (s12) edge (s01);
    \path ($(s20.0)-(0,0.5ex)$) edge ($(s21.150)-(0.3ex,0.5ex)$);
    \path ($(s21.150)+(0.3ex,0.5ex)$) edge ($(s20.0)+(0,0.5ex)$);
    \path (s21) edge (s11);
    \path (s22) edge (s21.30);
\end{tikzpicture}
\caption{Left: An example of a stochastic game. Middle: The same stochastic game with a priority function. Right: An example of strategies for players Even and Odd, and the winning sets of players Even (thick vertices) and Odd (dotted vertices) for the parity objective.}\label{sg}
\end{figure}

The game consists of three players: Even ($\square$), Odd ($\bigcirc$), and Random ($\mathlarger{\mathlarger{\triangle}}$). They take turns moving a token from vertex to vertex, forming a path. At a vertex in $V_\square$ (resp. $V_\bigcirc$), player Even (resp. Odd) moves the token to one of its successors. When the token is at a vertex in $V_\mathlarger{\triangle}$, player Random moves the token to one of its successors uniformly at random. We assume that there always exists at least one out-going edge at each vertex, implying that any path in the game can always be extended to an infinite path. An example of a stochastic game is illustrated in Fig.~\ref{sg}. Let $\mathcal{D}(V)$ denote the set of probability distributions on $V$. \emph{Strategies} for players Even and Odd are defined as follows.

\begin{definition}[Strategy]
A \emph{strategy for player Even} is $\sigma_\square : V^* \times V_\square \to \mathcal{D}(V)$ describing its next move. A \emph{strategy for player Odd} is $\sigma_\bigcirc : V^* \times V_\bigcirc \to \mathcal{D}(V)$.
\end{definition}

Intuitively, a strategy assigns the probability for a player to move to a successor vertex based on the path of previously visited vertices. Given a measurable set of infinite paths $P \subseteq V^\omega$, an initial vertex $v_0$, and a pair $(\sigma_\square,\sigma_\bigcirc)$ of strategies, the probability that an infinite path generated under $(\sigma_\square,\sigma_\bigcirc)$ belongs to $P$ is uniquely defined. We write $\Pr^{\sigma_\square,\sigma_\bigcirc}_{v_0}[P]$ for the probability that an infinite path belongs to $P$ if the game starts at $v_0$ and the players' strategies are $\sigma_\square$ and $\sigma_\bigcirc$.

We specify \emph{winning objectives} of the game using LTL formulae where atomic propositions are vertices (i.e., the set $\AP$ in Def.~\ref{defn:LTL} is $V$). For notational convenience, we write $\Pr^{\sigma_\square,\sigma_\bigcirc}_{v_0}[\varphi]$ for $\Pr^{\sigma_\square,\sigma_\bigcirc}_{v_0}[P_\varphi]$ where $P_\varphi = \{\bar{v} \in V^\omega : \bar{v} \vDash \varphi\}$. Given $X \subseteq V$, we focus on five winning objectives: \emph{safety} $\always X$ means a path always stays in $X$, \emph{reachability} $\eventually X$ means a path eventually reaches $X$, \emph{Büchi} $\buchi X$ means a path visits $X$ infinitely often, \emph{co-Büchi} $\cobuchi X$ means a path eventually stays in $X$, and \emph{parity}. For a parity objective, we are given a \emph{priority function} $p : V \to \{0, \ldots, d\}$ for some $d \in \mathbb{N}$. Let $V_i := \{v \in V : p(v) = i\}$. Then, the parity objective is $\bigwedge_{i \in \{1,3,\ldots,2\cdot\lceil d/2 \rceil-1\}} \left( \buchi V_i \implies \bigvee_{j \in \{0,2,\ldots,i-1\}} \buchi V_j \right)$. In other words, an infinite path satisfies the parity objective if the minimum priority seen infinitely often along the path is even. Figure~\ref{sg}(middle) shows an instance of a stochastic game with a parity objective, where the priority of each vertex is written inside that vertex.

Consider a winning objective $\varphi$, we say that a strategy $\sigma_\square$ of player Even is \emph{almost-sure winning} from $v_0$ if for all strategies $\sigma_\bigcirc$ of player Odd, we have $\Pr^{\sigma_\square,\sigma_\bigcirc}_{v_0}[\varphi] = 1$. A strategy $\sigma_\square$ of player Even is \emph{positive winning} from $v_0$ if for all strategies $\sigma_\bigcirc$ of player Odd, we have $\Pr^{\sigma_\square,\sigma_\bigcirc}_{v_0}[\varphi] > 0$. We denote $W_\square \subseteq V$, called the \emph{winning set of player Even}, for the set of vertices from which there exists an almost-sure winning strategy for player Even. Similarly, we denote $W_\bigcirc \subseteq V$, called the \emph{winning set of player Odd}, for the set of vertices from which there does not exist a positive winning strategy for player Even. Then, we are interested in the following problem.

\begin{definition}[Winning Strategy Computation]
Given an SG $G$ and a winning objective $\varphi$, the \emph{almost-sure winning strategy computation problem} is to compute a strategy $\sigma_\square$ of player Even that is almost-sure winning from all $v \in W_\square$. Analogously, the \emph{positive winning strategy computation problem} is to compute a strategy $\sigma_\square$ of player Even that is positive winning from all $v \in V \setminus W_\bigcirc$.
\end{definition}

Figure~\ref{sg}(right) provides an example of the winning sets $W_\square$ and $W_\bigcirc$ for the parity objective, together with an almost-sure and positive winning strategy $\sigma_\square$ of player Even. We note that, under the strategy $\sigma_\square$, an infinite path starting from any $v \in W_\square$ visits the vertex with priority $0$ infinitely often with probability $1$.

\subsection{Set Operators}
Using $\mu$-calculus, $\mu Y.f(Y)$ and $\nu Y.f(Y)$ denote the least and greatest fixed points of a function $f : 2^V \to 2^V$. They can be computed via Kleene's fixed point theorem (see e.g.,~\cite{DBLP:journals/dm/Baranga91}). For $X \subseteq V$, we define the following set operators.
\begin{itemize}
    \item $\textsc{Pre}(X) := \{u \in V : \forall v \in V, (u,v) \in E \implies v \in X\}$
    \item $\textsc{Pre}_\square(X) := \{u \in V_\square : \exists v \in V, (u,v) \in E \wedge v \in X\}$
    \item $\textsc{Pre}_\bigcirc(X) := \{u \in V_\bigcirc : \exists v \in V, (u,v) \in E \wedge v \in X\}$
    \item $\textsc{Attr}(X) := \mu Y.(X \cup \textsc{Pre}(Y))$
    \item $\textsc{Attr}_\square(X) := \mu Y.(X \cup \textsc{Pre}(Y) \cup \textsc{Pre}_\square(Y))$
    \item $\textsc{Attr}_\bigcirc(X) := \mu Y.(X \cup \textsc{Pre}(Y) \cup \textsc{Pre}_\bigcirc(Y))$
\end{itemize}

\noindent In brief, $\textsc{Pre}(X)$ contains vertices that must reach $X$ in one step, and $\textsc{Pre}_\square(X)$ (resp. $\textsc{Pre}_\bigcirc(X)$) contains player Even's (resp. player Odd's) vertices that can reach $X$ in one step. The $\textsc{Attr}$ operators are defined similarly but for reaching $X$ in finitely many steps. Additionally, we define more set operators inspired by Banerjee et al.~\cite{DBLP:conf/tacas/BanerjeeMMSS22}, where $X,X' \subseteq V$.
\begin{itemize}
    \item $\textsc{Pre}_\mathlarger{\triangle}(X',X) := \{ u \in V_\mathlarger{\triangle} : (\forall v \in V, (u,v) \in E \implies v \in X') \wedge (\exists v \in V, (u,v) \in E \wedge v \in X) \}$
    \item $\textsc{Attr}'(X) := \nu Z.\mu Y (X \cup \textsc{Pre}(Y) \cup \textsc{Pre}_\mathlarger{\triangle}(Z,Y))$
    \item $\textsc{Attr}_\square'(X) := \nu Z.\mu Y (X \cup \textsc{Pre}_\square(Y) \cup \textsc{Pre}(Y) \cup \textsc{Pre}_\mathlarger{\triangle}(Z,Y))$
    \item $\textsc{Attr}_\bigcirc'(X) := \nu Z.\mu Y (X \cup \textsc{Pre}_\bigcirc(Y) \cup \textsc{Pre}(Y) \cup \textsc{Pre}_\mathlarger{\triangle}(Z,Y))$
\end{itemize}
The set $\textsc{Pre}_\mathlarger{\triangle}(X',X)$ consists of player Random's vertices whose all edges lead to $X'$ and some edges lead to $X$. The operators $\textsc{Attr}'$, $\textsc{Attr}'_\square$, and $\textsc{Attr}'_\bigcirc$ are defined analogously to $\textsc{Attr}$, $\textsc{Attr}_\square$, and $\textsc{Attr}_\bigcirc$ respectively, accounting for player Random's vertices. With these set operators, we state below the result from~\cite{DBLP:conf/tacas/BanerjeeMMSS22}.

\begin{theorem}[{\cite[Thm.~3--4]{DBLP:conf/tacas/BanerjeeMMSS22}}]\label{thm:banerjee}
Given an SG $G = (V,E,(V_\square,V_\bigcirc,V_\mathlarger{\triangle}))$ and $X \subseteq V$. The set $\textnormal{\textsc{Attr}}'_\square(X)$ is the winning set of player Even for $\eventually X$. Furthermore, the set $\nu Z.\mu Y((X \cap \textnormal{\textsc{Pre}}_\square(Z) \cap \textnormal{\textsc{Pre}}(Z)) \cup \textnormal{\textsc{Pre}}_\square(Y) \cup \textnormal{\textsc{Pre}}(Y) \cup \textnormal{\textsc{Pre}}_\mathlarger{\triangle}(Z,Y))$ is the winning set of player Even for $\buchi X$.
\end{theorem}

\subsection{Solving Stochastic Parity Games}
The winning set of player Even for stochastic parity games can be compute\change{d} by reducing the games into deterministic parity games (i.e., no player Random) and then applying existing techniques for deterministic parity games.

The reduction from stochastic parity games to deterministic parity games was proposed by Chatterjee et al.~\cite{DBLP:conf/csl/ChatterjeeJH03}. Briefly, the reduction described in Alg.~\ref{algo_gadget} replaces each vertex of player Random with a \emph{gadget}, which is a directed graph consisting of three layers of vertices (Fig.~\ref{gadget}), \change{while preserving the game structure for vertices of players Even and Odd}. It was proved in~\cite{DBLP:conf/csl/ChatterjeeJH03} that if player Even wins at a vertex in the (reduced) deterministic parity game, then player Even almost-sure wins at the corresponding vertex in the stochastic parity game. Notice that, for vertices of the deterministic parity game where player Odd wins, player Even may or may not positive win at the corresponding vertex in the stochastic parity game.

\begin{algorithm}[t]
\DontPrintSemicolon
\SetAlgoNoLine
$\textsc{Reduce}(G = (V,E,(V_\square,V_\bigcirc,V_\mathlarger{\triangle})), p : V \to \{0, \ldots, d\})$\\ \Indp
    $V'_\square \gets \emptyset; V'_\bigcirc \gets \emptyset; E' \gets \emptyset$\\
    \lForEach{$v \in V_\square$}{$V'_\square \gets V'_\square \cup \{v'\}; p'(v') \gets p(v)$}
    \lForEach{$v \in V_\bigcirc$}{$V'_\bigcirc \gets V'_\bigcirc \cup \{v'\}; p'(v') \gets p(v)$}
    \ForEach{$v \in V_\mathlarger{\triangle}$}{
        $V'_\bigcirc \gets V'_\bigcirc \cup \{v'\}; p'(v') \gets p(v)$\\
        \lFor{$i \in \{0, \ldots, \lceil p(v)/2 \rceil\}$}{$V'_\square \gets V'_\square \cup \{v_i'\}; p'(v_i') \gets p(v); E' \gets E' \cup \{(v',v_i')\}$
        }
        \For{$j \in \{0, \ldots, p(v)\}$}{
            \lIf{\upshape $j$ is even}{$V'_\bigcirc \gets V'_\bigcirc \cup \{v_{\lceil j/2 \rceil,j}'\}$ \textbf{else} $V'_\square \gets V'_\square \cup \{v_{\lceil j/2 \rceil,j}'\}$}
            $p'(v_{\lceil j/2 \rceil,j}') \gets j; E' \gets E' \cup \{(v_{\lceil j/2 \rceil}',v_{\lceil j/2 \rceil,j}')\}$
        }
    }
    \ForEach{$(u,v) \in E$}{
        \lIf{$u \in V_\mathlarger{\triangle}$}{\textbf{for} $j \in \{0, \ldots, p(u)\}$ \textbf{do} $E' \gets E' \cup \{(u_{\lceil j/2 \rceil,j}',v')\}$}
        \lElse{$E' \gets E' \cup \{(u',v')\}$}
    }
    \Return $(G' = (V' = V'_\square \cup V'_\bigcirc, E', (V'_\square,V'_\bigcirc,\emptyset)),p')$
\caption{Reducing stochastic to deterministic parity games~\cite{DBLP:conf/csl/ChatterjeeJH03}.}\label{algo_gadget}
\end{algorithm}

\begin{figure}[t]
\centering
\begin{tikzpicture}[
    cir/.style={minimum size=0pt},
    box/.style={minimum size=25pt, regular polygon, regular polygon sides=4, inner sep=0pt}
]
    \node [state, cir, label=below:$v'_{0,0}$] at (   0, 0) (v00) {$0$};
    \node [state, box, label=below:$v'_{1,1}$] at ( 1.5, 0) (v11) {$1$};
    \node [state, cir, label=below:$v'_{1,2}$] at ( 2.5, 0) (v12) {$2$};
    \node [state, box, label=below:$v'_{2,3}$] at ( 3.5, 0) (v23) {$3$};
    \node [state, cir, label=below:$v'_{2,4}$] at ( 4.5, 0) (v24) {$4$};
    \node              at (5.75, 0) {$\cdots$};
    \node [state, box, label=below:$v'_{\lceil p(v)/2 \rceil,p(v)}$] at (   7, 0) (v99) {$\phantom{0}$};
    \node              at (   7, 0) {\small$p(v)$};

    \node [state, box, label=right:$v'_{0}$] at (  0, 1.2) (v0) {$\phantom{0}$};
    \node              at (  0, 1.2) {\small$p(v)$};
    \node [state, box, label=right:$v'_{1}$] at (  2, 1.2) (v1) {$\phantom{0}$};
    \node              at (  2, 1.2) {\small$p(v)$};
    \node [state, box, label=right:$v'_{2}$] at (  4, 1.2) (v2) {$\phantom{0}$};
    \node              at (  4, 1.2) {\small$p(v)$};
    \node              at (5.5, 1.2) {$\cdots$};
    \node [state, box, label=right:$v'_{\lceil p(v)/2 \rceil}$] at (  7, 1.2) (v9) {$\phantom{0}$};
    \node              at (  7, 1.2) {\small$p(v)$};

    \node [state, cir, label=right:$v'$] at (3.5, 2.4) (v) {$\phantom{0}$};
    \node              at (3.5, 2.4) {\small$p(v)$};

    \path (v.270) edge (v0.45);
    \path (v.270) edge (v1.45);
    \path (v.270) edge (v2);
    \path (v.270) edge (v9.135);
    \path (v0) edge (v00);
    \path (v1.270) edge (v11);
    \path (v1.270) edge (v12);
    \path (v2.270) edge (v23);
    \path (v2.270) edge (v24);
    \path (v9) edge (v99);
\end{tikzpicture}
\caption{Gadget of~\cite{DBLP:conf/csl/ChatterjeeJH03} for reducing stochastic parity games to deterministic parity games.}\label{gadget}
\end{figure}
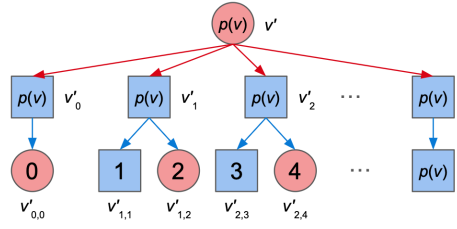

\begin{lemma}[{\cite[Lem. 3]{DBLP:conf/csl/ChatterjeeJH03}}]\label{lem:chatterjee}
Let $(G',p') \gets \textnormal{\textsc{Reduce}}(G,p)$. For every vertex $v$ in $G$, if player Even has a winning strategy from $v'$ in $G'$, then player Even has an almost-sure winning strategy from $v$.
\end{lemma}

To solve deterministic parity games, various algorithms can be used~\cite{DBLP:conf/stoc/CaludeJKL017,DBLP:journals/tcs/Zielonka98}. In this work, we mainly consider the recursive algorithm of Zielonka~\cite{DBLP:journals/tcs/Zielonka98} shown in Alg.~\ref{algo_zielonka}. The algorithm returns the winning sets of both players $(W_\square,W_\bigcirc)$. For $G = (V,E,(V_\square,V_\bigcirc,V_\mathlarger{\triangle}))$ and $X \subseteq V$, we use $G \setminus X$ as a shorthand for the subgame $(V \setminus X,E \setminus (X \times V \cup V \times X),(V_\square \setminus X,V_\bigcirc \setminus X,V_\mathlarger{\triangle} \setminus X))$.

\begin{algorithm}[t]
\DontPrintSemicolon
\SetAlgoNoLine
$\textsc{Solve}(G = (V,E,(V_\square,V_\bigcirc,\emptyset)), p : V \to \{0, \ldots, d\})$\\ \Indp
    \lIf{$V = \emptyset$}{\Return $(W_\square,W_\bigcirc) = (\emptyset,\emptyset)$}
    $x \gets \min \{p(v) : v \in V\}; X \gets \arg\min \{p(v) : v \in V\}$\\
    \If{\upshape $x$ is even}{
        $A \gets \textsc{Attr}_\square(X)$\\
        $(W'_\square,W'_\bigcirc) \gets \textsc{Solve}(G \setminus A, p)$\\
        \lIf{$W'_\bigcirc = \emptyset$}{\Return $(W_\square,W_\bigcirc) = (V,\emptyset)$}
        $B \gets \textsc{Attr}_\bigcirc(W'_\bigcirc)$\\
        $(W''_\square,W''_\bigcirc) \gets \textsc{Solve}(G \setminus B, p)$\\
        \Return $(W_\square,W_\bigcirc) = (W''_\square,W''_\bigcirc \cup B)$
    }
    \Else{
        $A \gets \textsc{Attr}_\bigcirc(X)$\\
        $(W'_\square,W'_\bigcirc) \gets \textsc{Solve}(G \setminus A, p)$\\
        \lIf{$W'_\square = \emptyset$}{\Return $(W_\square,W_\bigcirc) = (\emptyset,V)$}
        $B \gets \textsc{Attr}_\square(W'_\square)$\\
        $(W''_\square,W''_\bigcirc) \gets \textsc{Solve}(G \setminus B, p)$\\
        \Return $(W_\square,W_\bigcirc) = (W''_\square \cup B,W''_\bigcirc)$
    }
\caption{Solving deterministic parity games~\cite{DBLP:journals/tcs/Zielonka98}.}\label{algo_zielonka}
\end{algorithm}

\change{Zielonka's algorithm (Alg.~\ref{algo_zielonka}) follows a divide-and-conquer approach. It begins by identifying the set $X$ of vertices with the minimum priority. We assume that this minimum priority is even (the case where it is odd is handled analogously). The algorithm then computes the set $A$ from which player Even can force the play to reach $X$, and recursively solves the subgame $G \setminus A$. Let the resulting winning sets be $(W'_\square,W'_\bigcirc)$.

If $W'_\bigcirc = \emptyset$, then player Odd cannot win in $G$: any infinite path either eventually stays within $G \setminus A$ (where player Even wins), or visits $A$ infinitely often (ensuring the minimum priority---even---appears infinitely often).

Otherwise, the algorithm computes the set $B$ from which player Odd can force the play to reach $W'_\bigcirc$, and recursively solves the subgame $G \setminus B$. Let the resulting winning sets be $(W''_\square,W''_\bigcirc)$. In this case, player Even can control the play starting from $W''_\square$ to remain within $W''_\square$, avoiding both $W''_\bigcirc$ and $B$ (by the definition of $\textsc{Attr}_\bigcirc(W'_\bigcirc)$). Therefore, the winning set of player Even in $G$ is precisely $W''_\square$.}

\section{Almost-Sure Winning Strategy Templates}\label{sec3}
The concept of \emph{strategy templates} considered in this work was introduced in~\cite{DBLP:conf/tacas/AnandMNS23,DBLP:conf/cav/AnandNS23}. In this section, we redefine \emph{strategy templates} for our setting of stochastic games.

\begin{definition}[Strategy Template]\label{def7}
Given an SG $G = (V,E,(V_\square,V_\bigcirc,V_\mathlarger{\triangle}))$ and let $E_\square := E \cap (V_\square \times V)$, a \emph{strategy template} is $T = (P, \LL, C)$ where $P \subseteq E_\square$ is a set of \emph{prohibited edges}, $\LL \subseteq 2^{E_\square}$ is a set of \emph{live-groups}, and $C \subseteq E_\square$ is a set of \emph{co-live edges}.
\end{definition}

\begin{definition}[LTL Formula induced from Template]\label{def8}
Given a strategy template $T = (P, \LL, C)$, we define four LTL formulae induced from $T$ as follows.
\begin{itemize}
    \item $\psi_P := \bigwedge_{(u,v) \in P} \always(u \implies \neg \XX v)$,
    \item $\psi_\LL := \bigwedge_{L \in \LL}\left(\left(\bigvee_{(u,v) \in L}\buchi u\right) \implies \left(\bigvee_{(u,v) \in L}\buchi(u \wedge \XX v)\right)\right)$,
    \item $\psi_C := \bigwedge_{(u,v) \in C} \cobuchi(u \implies \neg\XX v)$,
    \item $\psi_T := \psi_P \wedge \psi_\LL \wedge \psi_C$.
\end{itemize}
\end{definition}

In brief, a strategy template $T = (P, \LL, C)$ describes a set of infinite paths with certain properties. Namely, infinite paths satisfying $\psi_P$ do not use edges in $P$; those satisfying $\psi_\LL$ have a property that: for each $L \in \LL$, if there is a vertex $u$ such that $(u,v) \in L$ and $u$ is visited infinitely often, then an edge in $L$ is used infinitely often; and those satisfying $\psi_C$ use edges in $C$ only finitely often.

Based on Def.~\ref{def7} and~\ref{def8}, we define \emph{almost-sure winning strategy templates} for SGs as follows. Recall that $W_\square \subseteq V$ denotes the winning set of player Even for the winning objective $\varphi$.

\begin{definition}[Almost-Sure Winning Strategy Template]\label{def:as_tem}
Given an SG $G$ and a winning objective $\varphi$, a strategy template $T = (P, \LL, C)$ is \emph{almost-sure winning for $\varphi$} if $\Pr^{\sigma_\square,\sigma_\bigcirc}_{v_0}[\psi_T] = 1 \implies \Pr^{\sigma_\square,\sigma_\bigcirc}_{v_0}[\varphi] = 1$, for any $v_0 \in W_\square$ and any pair of strategies $(\sigma_\square,\sigma_\bigcirc)$.
\end{definition}

It follows directly from Def.~\ref{def:as_tem} that a strategy $\sigma_\square$ of player Even is almost-sure winning if, under any strategy $\sigma_\bigcirc$ of player Odd, all generated paths from $W_\square$ satisfy $\psi_T$ of an almost-sure winning template $T$ with probability 1.

We now present algorithms to construct almost-sure winning strategy templates for five winning objectives---safety, reachability, Büchi, co-Büchi, and parity---and prove their correctness. For brevity, almost-sure winning strategy templates in this section are referred to by \emph{templates}. We let $\textsc{Edges}_\square(X,Y) := \{(u,v) \in E : u \in X \cap V_\square \wedge v \in Y\}$.

\subsection{Templates for Safety Objectives}
A safety objective is of the form $\always X$ where $X \subseteq V$. Algorithm~\ref{algo_safety} first computes $W_\square$ and then returns $T = (P,\emptyset,\emptyset)$. Since player Even must not leave $W_\square$, the set $P$ contains all player Even's edges that leave $W_\square$.

\begin{algorithm}[t]
\DontPrintSemicolon
\SetAlgoNoLine
$\textsc{SafetyTemplate}(G = (V,E,(V_\square,V_\bigcirc,V_\mathlarger{\triangle})), X \subseteq V)$\\ \Indp
    $W_\square \gets \nu Y.(X \cap (\textsc{Pre}_\square(Y) \cup \textsc{Pre}(Y)))$\label{algo_safety:line2}\\
    $P \gets \textsc{Edges}_\square(W_\square,V \setminus W_\square)$\label{algo_safety:line3}\\
    \Return $(P, \LL = \emptyset, C = \emptyset)$
\caption{Constructing templates for $\always X$.}\label{algo_safety}
\end{algorithm}

\begin{theorem}\label{thm1}
$\textnormal{\textsc{SafetyTemplate}}(G,X)$ is almost-sure winning for $\always X$.
\end{theorem}
\begin{proof}
The greatest fixed point in Line~\ref{algo_safety:line2} provides the winning set $W_\square$ for $\always X$. Then, the template is constructed as $T = (P,\emptyset,\emptyset)$ where $P = \textsc{Edges}_\square(W_\square,V \setminus W_\square)$. Consider any infinite path $\bar{v} = v_0 v_1 \ldots$ with $v_0 \in W_\square$ that satisfies $\psi_T$. For any $i \in \mathbb{N}$, if $v_i \in W_\square$ then $v_{i+1} \in W_\square$, as $P$ does not allow a path to use an edge $(v_i,v_{i+1})$ where $v_{i+1} \in V \setminus W_\square$. Therefore, by induction, $v_i \in W_\square$ for all $i \in \mathbb{N}$ and $\bar{v}$ satisfies $\always X$.
\end{proof}

\subsection{Templates for Reachability Objectives}
A reachability objective is of the form $\eventually X$ where $X \subseteq V$. Firstly, Alg.~\ref{algo_reachability} computes $A \gets \textsc{Attr}'(X)$, meaning that all infinite paths starting in $A$ eventually reach $X$. Then, it computes $W_\square \gets \textsc{Attr}'_\square(A)$. By definition, player Even can eventually reach $A$ from a vertex in $W_\square$ regardless of player Odd's strategy. Since player Even must neither leave $W_\square$ nor stay in $W_\square \setminus A$ infinitely often (before reaching $X$), the sets $P$ and $C$ are computed correspondingly.

\begin{algorithm}[t]
\DontPrintSemicolon
\SetAlgoNoLine
$\textsc{ReachabilityTemplate}(G = (V,E,(V_\square,V_\bigcirc,V_\mathlarger{\triangle})), X \subseteq V)$\\ \Indp
    $A \gets \textsc{Attr}'(X)$\\
    $W_\square \gets \textsc{Attr}'_\square(A)$\label{algo_reachability:line3}\\
    $P \gets \textsc{Edges}_\square(W_\square,V \setminus W_\square)$\\
    $C \gets \textsc{Edges}_\square(W_\square \setminus A, W_\square \setminus A)$\\
    \Return $(P, \LL = \emptyset, C)$
\caption{Constructing templates for $\eventually X$.}\label{algo_reachability}
\end{algorithm}

\begin{theorem}\label{thm2}
$\textnormal{\textsc{ReachabilityTemplate}}(G,X)$ is almost-sure winning for $\eventually X$.
\end{theorem}
\begin{proof}
By definition, $A$ is the largest possible set of vertices from which any infinite path reaches $X$ almost-surely, regardless of players' strategy. From Line~\ref{algo_reachability:line3} and Thm.~\ref{thm:banerjee} (\cite[Thm.~4]{DBLP:conf/tacas/BanerjeeMMSS22}), $W_\square$ is the winning set for $\eventually X$. Then, $T$ is $(P,\emptyset,C)$ where $P = \textsc{Edges}_\square(W_\square,V \setminus W_\square)$ and $C = \textsc{Edges}_\square(W_\square \setminus A, W_\square \setminus A)$. Consider any infinite path $\bar{v} = v_0 v_1 \ldots$ with $v_0 \in W_\square$ that satisfies $\psi_T$. If $v_0 \in A$, then it almost-surely reaches $X$. Otherwise, $v_0 \in W_\square \setminus A$. By constraints of $P$ and $C$, the path can neither leave $W_\square$ nor stay in $W_\square \setminus A$ infinitely often. Thus, the path must almost-surely reach $A$ and therefore $X$, satisfying $\eventually X$ with probability 1.
\end{proof}

\subsection{Templates for Büchi Objectives}
A Büchi objective is of the form $\buchi X$ where $X \subseteq V$. The set $W_\square$ can be described as a fixed point in Line~\ref{algo_buchi:line2} of Alg.~\ref{algo_buchi} \change{(following the result of Thm.~\ref{thm:banerjee})}. The set $P$ is again the set of all edges leaving $W_\square$. The function $\textsc{LiveGroups}(G,X)$ iteratively constructs $A \gets \textsc{Attr}'(X)$ and $X \gets A \cup \textsc{Pre}_\square(A)$. For each of player Even's vertices in $X \setminus A$, there must be an edge going to $A$. When a path arrives in $X \setminus A$, one of such edges must be used in order to go to $A$. And eventually, the path arrives in $X$. This results in the construction of the set $\LL$.

\begin{algorithm}[t]
\DontPrintSemicolon
\SetAlgoNoLine
$\textsc{BüchiTemplate}(G = (V,E,(V_\square,V_\bigcirc,V_\mathlarger{\triangle})), X \subseteq V)$\\ \Indp
    $W_\square \gets \nu Z.\mu Y((X \cap \textsc{Pre}_\square(Z) \cap \textsc{Pre}(Z)) \cup \textsc{Pre}_\square(Y) \cup \textsc{Pre}(Y) \cup \textsc{Pre}_\mathlarger{\triangle}(Z,Y))$\label{algo_buchi:line2}\\
    $P \gets \textsc{Edges}_\square(W_\square,V \setminus W_\square)$\\
    \Return $(P, \LL = \textsc{LiveGroups}(G,X \cap W_\square), C = \emptyset)$\\ \Indm
\;
$\textsc{LiveGroups}(G = (V,E,(V_\square,V_\bigcirc,V_\mathlarger{\triangle})), X \subseteq V)$\\ \Indp
    $\LL \gets \emptyset$\\
    \While{\textbf{\upshape true}}{
        $A \gets \textsc{Attr}'(X)$\label{algo_buchi:line9}\\
        $X \gets A \cup \textsc{Pre}_\square(A)$\label{algo_buchi:line10}\\
        \lIf{$X = A$}{\textbf{break}}
        $\LL \gets \LL \cup \{\textsc{Edges}_\square(X \setminus A, A)\}$\\
    }
    \Return $\LL$
\caption{Constructing templates for $\buchi X$.}\label{algo_buchi}
\end{algorithm}

\begin{theorem}
$\textnormal{\textsc{BüchiTemplate}}(G,X)$ is almost-sure winning for $\buchi X$.
\end{theorem}
\begin{proof}
Let $T = (P, \LL, \emptyset) \gets \textnormal{\textsc{BüchiTemplate}}(G,X)$. The set $W_\square$ in Line~\ref{algo_buchi:line2} is the winning set of $\buchi X$ by Thm.~\ref{thm:banerjee} (\cite[Thm.~3]{DBLP:conf/tacas/BanerjeeMMSS22}). The set $P$ is the set of edges leaving $W_\square$. Hence, it is sufficient to show that all paths $\bar{v} = v_0 v_1 \ldots$ starting at $v_0 \in W_\square$ and following $T$ visit $X \cap W_\square$ infinitely often. Consider $\textsc{LiveGroups}(G, X \cap W_\square)$. Let $A_i$ and $X_i$ be the sets $A$ and $X$ computed in the $i$-th iteration with $X_0 = X \cap W_\square$. By Lines~\ref{algo_buchi:line9}--\ref{algo_buchi:line10}, any path from $A_i$ reaches $X_{i-1}$ almost-surely, and player Even can move from $X_i \setminus A_i$ to $A_i$. Notice that $\textsc{LiveGroups}$ always terminates when $X_i = A_i = W_\square$, as $W_\square$ is the winning set. Since $\LL$ contains $\textsc{Edges}_\square(X_i\setminus A_i, A_i)$, when the infinite path $\bar{v}$ reaches $X_i \setminus A_i$, it cannot stay there forever due to the restriction of $\LL$. Thus, the path eventually reaches $A_i$ and then $X_{i-1}$, and by induction, reaches $X_0$. The path then either stays in $X_0$ or continues to any $X_i$ or $A_i$, in which case returns to $X_0$. Hence, the path visits $X_0 = X \cap W_\square$ infinitely often, satisfying $\buchi X$ almost-surely.
\end{proof}

\subsection{Templates for Co-Büchi Objectives}
A co-Büchi objective is of the form $\cobuchi X$ where $X \subseteq V$. Algorithm~\ref{algo_cobuchi} constructs an almost-sure winning strategy template with two main steps. First, the algorithm finds the set of vertices that can almost-surely satisfy $\always X$. Then, the algorithm computes the set of vertices that can almost-surely reach that set.

\begin{algorithm}[t]
\DontPrintSemicolon
\SetAlgoNoLine
$\textsc{Co-BüchiTemplate}(G = (V,E,(V_\square,V_\bigcirc,V_\mathlarger{\triangle})), X \subseteq V)$\\ \Indp
    $X \gets \nu Y.(X \cap (\textsc{Pre}_\square(Y) \cup \textsc{Pre}(Y)))$\label{algo_cobuchi:line2}\\
    $A \gets \textsc{Attr}'(X)$\label{algo_cobuchi:line3}\\
    $W_\square \gets \textsc{Attr}'_\square(A)$\\
    $P \gets \textsc{Edges}_\square(W_\square,V \setminus W_\square)$\\
    $C \gets \textsc{Edges}_\square(X, W_\square \setminus X) \cup \textsc{Edges}_\square(W_\square \setminus A, W_\square \setminus A)$\label{algo_cobuchi:line6}\\
    \Return $(P, \LL = \emptyset, C)$
\caption{Constructing templates for $\cobuchi X$.}\label{algo_cobuchi}
\end{algorithm}

\begin{theorem}
$\textnormal{\textsc{Co-BüchiTemplate}}(G,X)$ is almost-sure winning for $\cobuchi X$.
\end{theorem}
\begin{proof}
We can show, in the same manner as Thm.~\ref{thm2}, that $X$ in Line~\ref{algo_cobuchi:line2} becomes the winning set of $\always X$. For Lines~\ref{algo_cobuchi:line3}--\ref{algo_cobuchi:line6}, Alg.~\ref{algo_cobuchi} follows Alg.~\ref{algo_reachability}. Hence, by Thm.~\ref{thm2}, $T = (P,\emptyset,C)$ is almost-sure winning for $\eventually(\always X) = \cobuchi X$.
\end{proof}

\subsection{Templates for Parity Objectives}
A parity objective comes with a priority function $p : V \to \{0, \ldots, d\}$ for some $d \in \mathbb{N}$. We construct an almost-sure winning strategy template for a parity objective by (i) reducing a stochastic game to a deterministic game (i.e., with \textsc{Reduce} in Alg.~\ref{algo_gadget}), (ii) constructing a winning strategy template for the reduced deterministic game, and (iii) converting the template for the deterministic game into a template for the stochastic game.

The construction of a winning strategy template for a deterministic parity game was presented in~\cite{DBLP:conf/cav/AnandNS23}, detailed in Alg.~\ref{algo_detparity}. It extends Alg.~\ref{algo_zielonka} which solves deterministic parity games. Instead of returning the set $P$, the algorithm returns $W_\square$ and $W_\bigcirc$, which can then be used to construct $P$. The algorithm also utilizes \textsc{LiveGroups} from Alg.~\ref{algo_buchi}.

\change{In Alg.~\ref{algo_detparity}, the winning sets $(W_\square,W_\bigcirc)$ are computed in the same way as in Alg.~\ref{algo_zielonka}. We briefly review how the sets $\LL$ and $C$ are recursively constructed, focusing on the case where the minimum priority is even (the odd case is handled analogously). First, if $A = V$, then player Even must ensure that the play visits $X$ infinitely often. In this case, we set $\LL = \textsc{LiveGroups}(G,X)$. Otherwise, the algorithm solves the subgame $G \setminus A$, resulting in $(W'_\square,W'_\bigcirc,\LL',C')$.

If $W'_\bigcirc = \emptyset$, then player Even must not only satisfy the constraints of $\LL'$ and $C'$, but also ensure that the play visits $X$ infinitely often. Thus, we set $\LL = \LL' \cup \textsc{LiveGroups}(G,X)$ and $C = C'$. Otherwise, the algorithm computes the set $B$ and solves the subgame $G \setminus B$, resulting in $(W''_\square,W''_\bigcirc,\LL'',C'')$. As established in Alg.~\ref{algo_zielonka}, the winning set of player Even is $W''_\square$. Therefore, in this case, it is sufficient for player Even to follow $\LL''$ and $C''$ in $G$.}

Regarding the process of converting the template, we give the algorithm in Alg.~\ref{algo_parity}. Essentially, we remove from $\LL'$ and $C'$ all edges $(u',v')$ such that $u'$ is part of a gadget (i.e., there is no corresponding vertex $u$ in $G$). The result of this procedure is then an almost-sure winning strategy template for the parity objective.

\begin{algorithm}[t]
\DontPrintSemicolon
\SetAlgoNoLine
$\textsc{DetParityTemplate}(G = (V,E,(V_\square,V_\bigcirc,\emptyset)), p : V \to \{0, \ldots, d\})$\\ \Indp
    $x \gets \min \{p(v) : v \in V\}; X \gets \arg\min \{p(v) : v \in V\}$\\
    \If{\upshape $x$ is even}{
        $A \gets \textsc{Attr}_\square(X)$\\
        \lIf{$A = V$}{\Return $(W_\square,W_\bigcirc,\LL,C) = (V,\emptyset, \textsc{LiveGroups}(G,X),\emptyset)$}
        $(W'_\square,W'_\bigcirc,\LL',C') \gets \textsc{DetParityTemplate}(G \setminus A, p)$\\
        \lIf{$W'_\bigcirc = \emptyset$}{\Return $(W_\square,W_\bigcirc,\LL,C) = (V,\emptyset,\LL' \cup \textsc{LiveGroups}(G,X),C')$}
        $B \gets \textsc{Attr}_\bigcirc(W'_\bigcirc)$\\
        $(W''_\square,W''_\bigcirc,\LL'',C'') \gets \textsc{DetParityTemplate}(G \setminus B, p)$\\
        \Return $(W_\square,W_\bigcirc,\LL,C) = (W''_\square,W''_\bigcirc \cup B, \LL'', C'')$
    }
    \Else{
        $A \gets \textsc{Attr}_\bigcirc(X)$\\
        \lIf{$A = V$}{\Return $(W_\square,W_\bigcirc,\LL,C) = (\emptyset,V,\emptyset,\emptyset)$}
        $(W'_\square,W'_\bigcirc,\LL',C') \gets \textsc{DetParityTemplate}(G \setminus A, p)$\\
        \lIf{$W'_\square = \emptyset$}{\Return $(W_\square,W_\bigcirc,\LL,C) = (\emptyset,V,\emptyset,\emptyset)$}
        $\LL' \gets \LL' \cup \textsc{LiveGroups}(G,W'_\square)$\tcp*{Algorithm~\ref{algo_buchi}}
        $C' \gets C' \cup \textsc{Edges}_\square(W'_\square,V \setminus W'_\square)$\\
        $B \gets \textsc{Attr}_\square(W'_\square)$\\
        $(W''_\square,W''_\bigcirc,\LL'',C'') \gets \textsc{DetParityTemplate}(G \setminus B, p)$\\
        \Return $(W_\square,W_\bigcirc,\LL,C) = (W''_\square \cup B,W''_\bigcirc,\LL' \cup \LL'',C' \cup C'')$
    }
\caption{Constructing templates for deterministic parity games~\cite{DBLP:conf/cav/AnandNS23}.}\label{algo_detparity}
\end{algorithm}

\begin{algorithm}[t]
\DontPrintSemicolon
\SetAlgoNoLine
$\textsc{ParityTemplate}(G = (V,E,(V_\square,V_\bigcirc,V_\mathlarger{\triangle})), p : V \to \{0, \ldots, d\})$\\ \Indp
    $(G',p') \gets \textsc{Reduce}(G,p)$\tcp*{Algorithm~\ref{algo_gadget}}
    $(W'_\square, W'_\bigcirc, \LL', C') \gets \textsc{DetParityTemplate}(G',p')$\tcp*{Algorithm~\ref{algo_detparity}}
    $P \gets \{ (u,v) \in E : (u',v') \in \textsc{Edges}_\square(W'_\square,W'_\bigcirc)\}$\tcp*{$u'$ in $G'$ corresponds to $u$ in $G$}
    $\LL \gets \emptyset; C \gets \emptyset$\\
    \ForEach{$L' \in \LL'$}{
        $L \gets \emptyset$\\
        \ForEach{$(u,v) \in E \cap (V_\square \times V)$}{
            \lIf{$(u',v') \in L'$}{$L \gets L \cup \{(u,v)\}$}
        }
        $\LL \gets \LL \cup \{L\}$
    }
    \ForEach{$(u,v) \in E \cap (V_\square \times V)$}{
        \lIf{$(u',v') \in C'$}{$C \gets C \cup \{(u,v)\}$}
    }
    \Return $(P, \LL, C)$
\caption{Constructing templates for stochastic parity games.}\label{algo_parity}
\end{algorithm}

\begin{theorem}
$\textnormal{\textsc{ParityTemplate}}(G,p)$ is almost-sure winning for a parity objective on $p$.
\end{theorem}
\begin{proof}
Let $(G',p') \gets \textsc{Reduce}(G,p)$ and Alg.~\ref{algo_detparity} returns $(W'_\square,W'_\bigcirc,\LL',C')$, following \change{the first two} lines. Then, $T' = (P' = \textsc{Edges}_\square(W'_\square,W'_\bigcirc),\LL',C')$ is winning for the parity objective of $(G',p')$ due to~\cite[Thm.~4]{DBLP:conf/cav/AnandNS23}. Moreover, it is shown in the proof of Lem.~\ref{lem:chatterjee} (\cite[Lem. 3]{DBLP:conf/csl/ChatterjeeJH03}) that $W'_\square$ and any winning strategy $\sigma'_\square$ for the parity objective of $(G',p')$ can be converted to $W_\square$ and an almost-sure winning strategy $\sigma_\square$ for the parity objective of $(G,p)$ by removing all vertices and edges introduced by gadgets. In a similar manner, a winning template $T'$ for the parity objective of $(G',p')$ can be converted to an almost-sure winning template $T$ for the parity objective of $(G,p)$ by removing all edges introduced by gadgets. This is exactly Alg.~\ref{algo_parity}. Thus, $\textsc{ParityTemplate}(G,p)$ is almost-sure winning for the parity objective on $p$.
\end{proof}

\subsection{Composing Almost-Sure Winning Strategy Templates}

We end this section with a remark on composing templates, discussed in~\cite{DBLP:conf/cav/AnandNS23}. In short, we can compose two almost-sure winning strategy templates $T = (P,\LL,C)$ for a winning objective $\varphi$ and $T' = (P',\LL',C')$ for a winning objective $\varphi'$ into $T'' = (P'', \LL \cup \LL', C \cup C')$ for the winning objective $\varphi \wedge \varphi'$, \change{where $P''$ is defined in Line \ref{alg9:line2} of Alg.~\ref{algo_compose}. This composition method works} unless $T$ and $T'$ \emph{conflict} (defined below). Otherwise, the conflict requires us to recompute an almost-sure winning strategy template for $\varphi \wedge \varphi'$ from scratch. This concept was formalized and proved by~\cite{DBLP:conf/cav/AnandNS23}, as shown in Alg.~\ref{algo_compose} and Thm.~\ref{thm:compose}.

\begin{definition}[Conflict-Freeness, {\cite[Def.~1]{DBLP:conf/cav/AnandNS23}}]
Given an SG $G$, a strategy template $T = (P, \LL, C)$ is \emph{conflict-free} if both following conditions are true: (i) for all $u \in V$, $\{(u,v) \in E\} \not\subseteq P \cup C$, and (ii) for all $u \in V$ and $L \in \LL$, $\{(u,v) \in L\} \neq \emptyset \implies \{(u,v) \in L\} \not\subseteq P \cup C$.
\end{definition}

\begin{theorem}[{\cite[Thm.~5]{DBLP:conf/cav/AnandNS23}}]\label{thm:compose}
Given an SG $G$, a pair $(W_\square,T)$ for a winning objective $\varphi$, and a pair $(W'_\square,T')$ for a winning objective $\varphi'$. Let $(W''_\square,T'') \gets \textnormal{\textsc{ComposeAlmostSure}}(G,\varphi,W_\square,T,\varphi',W'_\square,T')$. Then, $T''$ is a strategy template that is almost-sure winning with respect to the winning set $W''_\square$ of player Even for the winning objective $\varphi \wedge \varphi'$.
\end{theorem}

\begin{algorithm}[t]
\DontPrintSemicolon
\SetAlgoNoLine
$\textsc{ComposeAlmostSure}(G = (V,E,(V_\square,V_\bigcirc,V_\mathlarger{\triangle})), \varphi, W_\square, T = (P,\LL,C), \varphi', W'_\square, T' = (P',\LL',C'))$\\ \Indp
    $W''_\square \gets W_\square \cap W'_\square; P'' \gets \{(u,v) \in E : u \in W''_\square \cap V_\square \wedge v \not\in W''_\square\}; \LL'' \gets \LL \cup \LL'; C'' \gets C \cup C'$\label{alg9:line2}\\
    $V_\textnormal{conflict} \gets \{u \in W''_\square : (\{u\} \times W''_\square) \cap E \subseteq C''\} \cup \{u \in W''_\square : \exists L'' \in \LL'', \emptyset \subsetneq (\{u\} \times W''_\square) \cap L'' \subseteq C''\}$\\
    \lIf{$V_\textnormal{conflict} = \emptyset$}{\Return $(W''_\square, T'' = (P'', \LL'', C''))$}
    \Return $\textsc{ConstructTemplate}(G, \varphi \wedge \varphi' \wedge \bigwedge_{u \in V_\textnormal{conflict}} \cobuchi(\neg u))$\tcp*{recompute from scratch}
\caption{Composing almost-sure winning strategy templates~\cite[Alg. 4]{DBLP:conf/cav/AnandNS23}.}\label{algo_compose}
\end{algorithm}

\change{To compose two strategy templates $T$ and $T'$, Alg.~\ref{algo_compose} first computes the winning set of player Even for the objective $\varphi \wedge \varphi'$ as $W''_\square = W_\square \cap W'_\square$. Consequently, the set of prohibited edges $P''$ consists of edges of player Even that leave $W''_\square$. The sets $\LL''$ and $C''$ are composed as $\LL'' = \LL \cup \LL'$ and $C'' = C \cup C'$.

This composition works provided that the resulting strategy template $T''$ is conflict-free. Conflicts can arise in two cases: (i) there exists a vertex $u \in W''_\square$ from which all out-going edges are either prohibited or co-live (i.e., no edge from $u$ can be used infinitely often), or (ii) there exists a vertex $u \in W''_\square$ and a live-group $L'' \in \LL''$ such that all out-going edges from $u$ within $L''$ are prohibited or co-live. In either case, player Even cannot satisfy the template $T''$ if the play visits $u$ infinitely often. Therefore, if such vertices exist---identified by Alg.~\ref{algo_compose} as $V_\textnormal{conflict}$---a new strategy template must be constructed from scratch for the objective $\varphi \wedge \varphi'$ while ensuring that vertices in $V_\textnormal{conflict}$ are not visited infinitely often.}

\section{Positive Winning Strategy Templates}\label{sec99}
So far, we only consider almost-sure winning strategies and templates for those vertices in the winning set $W_\square$ of player Even. It is obvious that player Even has no chance of winning from vertices in the winning set $W_\bigcirc$ of player Odd. Nevertheless, for vertices in the set $W_? := V \setminus (W_\square \cup W_\bigcirc)$, player Even has a non-zero probability of winning if player Even chooses a correct successor. Precisely, it has to follow a \emph{positive winning} strategy. As an example, the bottom-left vertex in Fig.~\ref{sg}(middle) belongs to player Even. If it chooses to move to the right, the winning probability is zero. On the other hand, moving up gives the winning probability of 0.5. To capture positive winning strategies, we develop \emph{positive winning strategy templates}.

\begin{definition}[Positive Winning Strategy Template]\label{def:pos_tem}
Given an SG $G$ and a winning objective $\varphi$, a strategy template $T = (P, \LL, C)$ is \emph{positive winning for $\varphi$} if $\Pr^{\sigma_\square,\sigma_\bigcirc}_{v_0}[\psi_T] = 1 \implies \Pr^{\sigma_\square,\sigma_\bigcirc}_{v_0}[\varphi] > 0$, for any $v_0 \in W_\square \cup W_?$ and any pair of strategies $(\sigma_\square,\sigma_\bigcirc)$.
\end{definition}

Analogous to almost-sure winning strategy templates, it follows straightforwardly from Def.~\ref{def:pos_tem} that a strategy $\sigma_\square$ of player Even is positive winning if, under any strategy $\sigma_\bigcirc$ of player Odd, all generated paths from $W_\square \cup W_?$ satisfy $\psi_T$ of a positive winning template $T$ with probability 1.

Note that a strategy $\sigma_\square$ of player Even which follows a positive winning strategy template, however, does not guarantee to satisfy the winning objective with the maximum (i.e., optimal) probability. \change{Constructing optimal strategies for stochastic games typically involves intricate procedures, even for simple objectives like reachability~\cite{DBLP:conf/cav/PhalakarnTHH20,DBLP:conf/lics/KretinskyMW23}. To the best of our knowledge, while optimal strategies for parity objectives are known to exist~\cite{DBLP:conf/csl/ChatterjeeJH03}, we are not aware of any algorithm in the literature that explicitly constructs such strategies. Moreover, composing templates while preserving optimality is challenging, as it requires addressing multi-objective optimization.}

Also notice that a positive winning strategy template does not have to be almost-sure winning. Nonetheless, it is possible that a template is both almost-sure and positive winning.

\subsection{Constructing Positive Winning Strategy Templates}

We propose Alg.~\ref{algo_positive} for constructing a positive winning strategy template for a given winning objective. It first computes the winning sets $W_\square,W_\bigcirc$ of both players and also an almost-sure winning strategy template $T_{(a)} = (P_{(a)},\LL_{(a)}, C_{(a)})$ for the winning objective. The set $W_\bigcirc$ can be computed in the same way as $W_\square$ in Alg.~\ref{algo_safety}--\ref{algo_parity}. Then, the algorithm imposes two restrictions on $W_?$: (i) edges from $W_?$ to $W_\bigcirc$ must not be used, and (ii) edges within $W_?$ must be used only finitely often. With additional constraints added, the algorithm returns a positive winning strategy template $T_{(p)}$. By construction, the template $T_{(p)}$ is almost-sure winning as well.

\begin{algorithm}[t]
\DontPrintSemicolon
\SetAlgoNoLine
$\textsc{PositiveTemplate}(G = (V,E,(V_\square,V_\bigcirc,V_\mathlarger{\triangle})), \varphi)$\\ \Indp
    Compute $W_\square,W_\bigcirc$, and an almost-sure winning $T_{(a)} = (P_{(a)},\LL_{(a)},C_{(a)})$ for $\varphi$\label{alg10:line2}\tcp*{Algorithms~\ref{algo_safety}-\ref{algo_parity}}
    $W_? \gets V \setminus (W_\square \cup W_\bigcirc)$\\
    $P_{(p)} \gets P_{(a)} \cup \textsc{Edges}_\square(W_?,W_\bigcirc); \LL_{(p)} \gets \LL_{(a)}; C_{(p)} \gets C_{(a)} \cup \textsc{Edges}_\square(W_?,W_?)$\\
    \Return $T_{(p)} = (P_{(p)}, \LL_{(p)}, C_{(p)})$
\caption{Constructing positive winning strategy templates.}\label{algo_positive}
\end{algorithm}

\begin{theorem}\label{thm:pos}
$\textnormal{\textsc{PositiveTemplate}}(G,\varphi)$ is almost-sure and positive winning for $\varphi$.
\end{theorem}
\begin{proof}
From Line~\ref{alg10:line2} of Alg.~\ref{algo_positive}, $T_{(a)}$ is almost-sure winning with respect to $W_\square$. This implies that $T_{(a)}$ is positive winning with respect to $W_\square$. Hence, we only need to show that $T_{(p)}$ is positive winning with respect to $W_?$. We assume that player Odd plays optimally as this gives minimum winning probability for player Even. Then, we argue that an infinite path $\bar{v} = v_0 v_1 \ldots$ that begins at $v_0 \in W_?$ and follows $T_{(p)}$ must eventually visit $W_\square \cup W_\bigcirc$.

For the sake of contradiction, suppose $\bar{v}$ stays in $W_?$ with probability 1. Then, some subset $W_\ast$ of $W_?$ is visited infinitely often. Due to $C_{(p)}$, we have $W_\ast \cap V_\square = \emptyset$. Hence, $W_\ast \subseteq V_\bigcirc \cup V_\mathlarger{\triangle}$. However, since player Odd plays optimally and it chooses to stay in $W_\ast$, $W_\ast \subseteq W_?$ must be almost-sure winning for player Odd. This implies $W_\ast \subseteq W_\bigcirc$, contradicting the fact that $W_\bigcirc \cap W_? = \emptyset$. So, we conclude that $\bar{v}$ eventually exits $W_?$ and reaches $W_\square \cup W_\bigcirc$.

If $\bar{v}$ reaches $W_\bigcirc$ with probability 1, then $W_?$ must again be almost-sure winning for player Odd. Therefore, $\bar{v}$ must reach $W_\bigcirc$ with probability less than 1, implying that $\bar{v}$ reaches $W_\square$ with non-zero probability.
\end{proof}

\subsection{Composing Positive Winning Strategy Templates}

From Alg.~\ref{algo_positive}, one can see that a positive winning strategy template $T_{(p)}$ can be constructed when $W_\square$, $W_\bigcirc$, and an almost-sure winning strategy template $T_{(a)}$ are given. To compose positive winning strategy templates, one can compose the winning sets and the almost-sure templates, and then construct a positive winning strategy template accordingly. This procedure is formalized in Alg.~\ref{algo_compose_pos}. \change{Specifically, the algorithm first computes the winning sets of players Even and Odd for the objective $\varphi \wedge \varphi'$, from which the set $W''_?$ is immediately derived. Next, the sets $P''_{(a)}$, $\LL''_{(a)}$, and $C''_{(a)}$ are constructed in the same manner as in Alg.~\ref{algo_compose} for the almost-sure winning strategy templates composition. Subsequently, the positive winning strategy template---consisting of $P''_{(p)}$, $\LL''_{(p)}$, and $C''_{(p)}$---is constructed based on $(W''_?,W''_\bigcirc,P''_{(a)},\LL''_{(a)},C''_{(a)})$, using the positive winning strategy template construction (Alg.~\ref{algo_positive}). Finally, the conflict-freeness of $T''_{(a)}$ is verified using the same procedure as in Alg.~\ref{algo_compose}.}

\begin{algorithm}[t]
\DontPrintSemicolon
\SetAlgoNoLine
$\textsc{ComposePositive}(G, \varphi, W_\square, W_\bigcirc, T_{(a)} = (P_{(a)},\LL_{(a)},C_{(a)}), \varphi', W'_\square, W'_\bigcirc, T'_{(a)} = (P'_{(a)},\LL'_{(a)},C'_{(a)}))$\\ \Indp
    $W''_\square \gets W_\square \cap W'_\square; W''_\bigcirc \gets \textsc{Attr}'_\bigcirc(W_\bigcirc \cup W'_\bigcirc); W''_? \gets V \setminus (W''_\square \cup W''_\bigcirc)$\\
    $P''_{(a)} \gets \{(u,v) \in E : u \in W''_\square \cap V_\square \wedge v \not\in W''_\square\}; \LL''_{(a)} \gets \LL_{(a)} \cup \LL'_{(a)}; C''_{(a)} \gets C_{(a)} \cup C'_{(a)}$\\
    $P''_{(p)} \gets P''_{(a)} \cup \textsc{Edges}_\square(W''_?,W''_\bigcirc); \LL''_{(p)} \gets \LL''_{(a)}; C''_{(p)} \gets C''_{(a)} \cup \textsc{Edges}_\square(W''_?,W''_?)$\\
    $V_\textnormal{conflict} \gets \{u \in W''_\square : (\{u\} \times W''_\square) \cap E \subseteq C''_{(a)}\} \cup \{u \in W''_\square : \exists L''_{(a)} \in \LL''_{(a)}, \emptyset \subsetneq (\{u\} \times W''_\square) \cap L''_{(a)} \subseteq C''_{(a)}\}$\\
    \lIf{$V_\textnormal{conflict} = \emptyset$}{\Return $(W''_\square, W''_\bigcirc, T''_{(a)} = (P''_{(a)}, \LL''_{(a)}, C''_{(a)}), T''_{(p)} = (P''_{(p)}, \LL''_{(p)}, C''_{(p)}))$}
    \Return $\textsc{PositiveTemplate}(G, \varphi \wedge \varphi' \wedge \bigwedge_{u \in V_\textnormal{conflict}} \cobuchi(\neg u))$\tcp*{recompute from scratch}
\caption{Composing positive winning strategy templates.}\label{algo_compose_pos}
\end{algorithm}

\begin{theorem}
Given an SG $G$, a tuple $(W_\square,W_\bigcirc,T_{(a)})$ for a winning objective $\varphi$, and a tuple $(W'_\square,W'_\bigcirc,T'_{(a)})$ for a winning objective $\varphi'$. Let $(W''_\square, W''_\bigcirc,T''_{(a)},T''_{(p)}) \gets \textnormal{\textsc{ComposePositive}}(G,\varphi,W_\square,W_\bigcirc,T_{(a)},\varphi',W'_\square,W'_\bigcirc,T'_{(a)})$. Then, $T''_{(p)}$ is a strategy template that is positive winning with respect to the set $V \setminus W''_\bigcirc$ for the winning objective $\varphi \wedge \varphi'$.
\end{theorem}
\begin{proof}
In a case of conflict (i.e., $V_\textnormal{conflict} \neq \emptyset$), a positive winning template is recomputed from scratch, giving a correct output by Thm.~\ref{thm:pos}. Thus, we assume $V_\textnormal{conflict} = \emptyset$. By Thm.~\ref{thm:compose}, $W''_\square$ is the winning set of player Even and $T''_{(a)}$ is almost-sure winning for $\varphi \wedge \varphi'$. If $W''_\bigcirc$ is also the winning set of player Odd for $\varphi \wedge \varphi'$, then $T''_{(p)}$ is positive winning by Thm.~\ref{thm:pos}. Hence, it remains to prove that $W''_\bigcirc = \textnormal{\textsc{Attr}}'_\bigcirc(W_\bigcirc \cup W'_\bigcirc)$ is the winning set of player Odd for $\varphi \wedge \varphi'$.

Without loss of generality, we assume that both players play optimally. It is clear that if $v \in \textnormal{\textsc{Attr}}'_\bigcirc(W_\bigcirc \cup W'_\bigcirc)$, then $\Pr^{\sigma_\square,\sigma_\bigcirc}_v[\varphi \wedge \varphi'] = 0$, which means $v$ is in the winning set of player Odd. For the other direction, we show that if $v \not\in \textnormal{\textsc{Attr}}'_\bigcirc(W_\bigcirc \cup W'_\bigcirc)$, then $\Pr^{\sigma_\square,\sigma_\bigcirc}_v[\varphi \wedge \varphi'] \neq 0$. If $v \in W''_\square$, we have $\Pr^{\sigma_\square,\sigma_\bigcirc}_v[\varphi \wedge \varphi'] = 1$ by Thm.~\ref{thm:compose}.

The only case left is $v \in V \setminus (W''_\square \cup \textnormal{\textsc{Attr}}'_\bigcirc(W_\bigcirc \cup W'_\bigcirc))$, in which we show that an infinite path $\bar{v} = v_0 v_1 \ldots$ starting at $v_0 = v$ must eventually reach $W''_\square \cup \textnormal{\textsc{Attr}}'_\bigcirc(W_\bigcirc \cup W'_\bigcirc)$. For the purpose of contradiction, suppose $\bar{v}$ stays in $V \setminus (W''_\square \cup \textnormal{\textsc{Attr}}'_\bigcirc(W_\bigcirc \cup W'_\bigcirc))$ with probability 1. Then, there must be a subset $W_\ast$ that is visited infinitely often by the path. As both players play optimally, the probabilities $\Pr^{\sigma_\square,\sigma_\bigcirc}_v[\varphi]$ and $\Pr^{\sigma_\square,\sigma_\bigcirc}_v[\varphi']$ can be determined as either 0 or 1 for all $v \in W_\ast$. If $\Pr^{\sigma_\square,\sigma_\bigcirc}_v[\varphi]$ (resp. $\Pr^{\sigma_\square,\sigma_\bigcirc}_v[\varphi']$) is 0, then $v$ must be in $W_\bigcirc$ (resp. $W'_\bigcirc$). If $\Pr^{\sigma_\square,\sigma_\bigcirc}_v[\varphi] = \Pr^{\sigma_\square,\sigma_\bigcirc}_v[\varphi'] = 1$, then $v$ must be in $W''_\square$. In both cases, a contradiction occurs.

Now, we have that the path $\bar{v}$ must almost-surely arrives at $W''_\square \cup \textnormal{\textsc{Attr}}'_\bigcirc(W_\bigcirc \cup W'_\bigcirc)$. If $\bar{v}$ reaches $\textnormal{\textsc{Attr}}'_\bigcirc(W_\bigcirc \cup W'_\bigcirc)$ with probability 1, then $v$ must be in $\textnormal{\textsc{Attr}}'_\bigcirc(W_\bigcirc \cup W'_\bigcirc)$, leading to a contradiction. Therefore, $\bar{v}$ must reach $W''_\square$ with non-zero probability and $\Pr^{\sigma_\square,\sigma_\bigcirc}_v[\varphi \wedge \varphi'] \neq 0$.
\end{proof}

\section{Templates Comparison: Permissiveness and Sizes}\label{sec6}

Two properties of strategy templates, namely \emph{permissiveness} and \emph{sizes}, can be seen related. Intuitively, small templates are more permissive as they impose less restrictions. However, this is not always true. We explore these concepts in this section.

\subsection{Permissiveness of Strategy Templates}

In order to define the \emph{permissiveness} of strategy templates, we first introduce the definition of the \emph{language generated by an LTL formula} as follows.

\begin{definition}[Language generated by LTL Formula]
Given an SG $G$ and an LTL formula $\psi$, the \emph{language generated by $\psi$} is $\mathcal{L}(\psi) = \{ \bar{v} = v_0 v_1 \ldots : \bar{v} \vDash \psi \wedge \forall i \in \mathbb{N}, (v_i,v_{i+1}) \in E\}$.
\end{definition}

We now state the definition of the \emph{permissiveness} of strategy templates. Note that our definition follows the concept of permissiveness as introduced by~\cite{ramadge1987supervisory}.

\begin{definition}[Permissiveness of Strategy Template]\label{def:permissiveness}
Given an SG $G$ and two strategy templates $T = (P, \LL, C)$, $T' = (P', \LL', C')$, we say that $T$ is \emph{no more permissive than} $T'$ if $\mathcal{L}(\psi_T) \subseteq \mathcal{L}(\psi_{T'})$.
\end{definition}

When taking the winning criterion and objective into consideration, what we would like to have is a most permissive template that satisfies the winning criterion and objective. From our definition, we prove two propositions.

\begin{proposition}\label{prop2}
Given an SG $G$ and two strategy templates $T = (P, \LL, C)$ and $T' = (P', \LL', C')$, $T$ is no more permissive than $T'$ if $P \supseteq P'$, $\LL \supseteq \LL'$, and $C \supseteq C'$.
\end{proposition}
\begin{proof}
If $P \supseteq P'$, then $\mathcal{L}(\psi_P) \subseteq \mathcal{L}(\psi_{P'})$. Similarly, $\LL \supseteq \LL'$ and $C \supseteq C'$ imply $\mathcal{L}(\psi_\LL) \subseteq \mathcal{L}(\psi_{\LL'})$ and $\mathcal{L}(\psi_C) \subseteq \mathcal{L}(\psi_{C'})$, respectively. As a result, we obtain $\mathcal{L}(\psi_T = \psi_P \wedge \psi_\LL \wedge \psi_C) = \mathcal{L}(\psi_P) \cap \mathcal{L}(\psi_\LL) \cap \mathcal{L}(\psi_C) \subseteq \mathcal{L}(\psi_{P'}) \cap \mathcal{L}(\psi_{\LL'}) \cap \mathcal{L}(\psi_{C'}) = \mathcal{L}(\psi_{T'} = \psi_{P'} \wedge \psi_{\LL'} \wedge \psi_{C'})$.
\end{proof}

\begin{proposition}
Given an SG $G$, a winning objective $\varphi$, and two conflict-free strategy templates $T = (P, \LL, C)$ and $T' = (P', \LL', C')$. If $T$ is no more permissive than $T'$ and $T'$ is almost-sure (resp. positive) winning, then $T$ is almost-sure (resp. positive) winning.
\end{proposition}
\begin{proof}
We provide the proof for the case of almost-sure winning. The case of positive winning is similar.

Let $\mathcal{L}(\psi_T) \subseteq \mathcal{L}(\psi_{T'})$. Then, $\Pr^{\sigma_\square,\sigma_\bigcirc}_{v_0}[\psi_T] = 1 \implies \Pr^{\sigma_\square,\sigma_\bigcirc}_{v_0}[\psi_{T'}] = 1$, for any $v_0 \in W_\square$ and any pair of strategies $(\sigma_\square,\sigma_\bigcirc)$. Since $T'$ is almost-sure winning, we have, by Def.~\ref{def:as_tem}, $\Pr^{\sigma_\square,\sigma_\bigcirc}_{v_0}[\psi_{T'}] = 1 \implies \Pr^{\sigma_\square,\sigma_\bigcirc}_{v_0}[\varphi] = 1$. Therefore, we obtain $\Pr^{\sigma_\square,\sigma_\bigcirc}_{v_0}[\psi_{T}] = 1 \implies \Pr^{\sigma_\square,\sigma_\bigcirc}_{v_0}[\varphi] = 1$ as required.
\end{proof}

\subsection{Sizes of Strategy Templates}

We consider another property of a strategy template which is its \emph{size}. It is formally defined as follows.

\begin{definition}[Size of Template]
Given two templates $T = (P, \LL, C)$ and $T' = (P', \LL', C')$. The \emph{overall size} of $T$ is $|T| = |P|+\sum_{L \in \LL}|L|+|C|$ and the \emph{element-wise size} of $T$ is the tuple $\|T\| = (|P|,\sum_{L \in \LL}|L|,|C|)$. We say that $T$ is \emph{no overall larger than} $T'$ if $|T| \leq |T'|$. Also, we say that $T$ is \emph{no element-wise larger} than $T'$ if $|P| \leq |P'|$, $|\LL| \leq |\LL'|$, and $|C| \leq |C'|$.
\end{definition}

Proposition~\ref{prop2} implies that a template of smaller size can possibly be more permissive. However, the following example shows templates that are of different sizes but are equally permissive.

\begin{example*}
Figure~\ref{ex} provides three almost-sure winning strategy templates for $\eventually w$: $T_i = (\emptyset,\emptyset,C_i)$ for $i \in \{1,2,3\}$ with $C_1 = \{(u,v),(v,u)\}$, $C_2 = \{(u,v)\}$, and $C_3 = \{(v,u)\}$. We have that $\mathcal{L}(\psi_{T_1}) = \mathcal{L}(\psi_{T_2}) = \mathcal{L}(\psi_{T_3})$, but $C_2 \subsetneq C_1$ and $C_3 \subsetneq C_1$.

\begin{figure}[t]
\centering
\begin{tikzpicture}[
    cir/.style={minimum size=0pt},
    box/.style={minimum size=25pt, regular polygon, regular polygon sides=4, inner sep=0pt}
]
    \node [state, box] at (  0, 0) (u) {$u$};
    \node [state, box] at (1.4, 0) (v) {$v$};
    \node [state, cir] at (2.8, 0) (w) {$w$};

    \node              at (1.4, 1) {$T_1$};

    \path ($(u.0)+(0,0.5ex)$) edge node[above]{$C$} ($(v.180)+(0,0.5ex)$);
    \path ($(v.180)-(0,0.5ex)$) edge node[below]{$C$} ($(u.0)-(0,0.5ex)$);
    \path ($(v.0)+(0,0.5ex)$) edge ($(w.180)+(0,0.5ex)$);
    \path ($(w.180)-(0,0.5ex)$) edge ($(v.0)-(0,0.5ex)$);
\end{tikzpicture}\qquad\qquad\begin{tikzpicture}[
    cir/.style={minimum size=0pt},
    box/.style={minimum size=25pt, regular polygon, regular polygon sides=4, inner sep=0pt}
]
    \node [state, box] at (  0, 0) (u) {$u$};
    \node [state, box] at (1.4, 0) (v) {$v$};
    \node [state, cir] at (2.8, 0) (w) {$w$};

    \node              at (1.4, 1) {$T_2$};

    \path ($(u.0)+(0,0.5ex)$) edge node[above]{$C$} ($(v.180)+(0,0.5ex)$);
    \path ($(v.180)-(0,0.5ex)$) edge node[below]{$\phantom{C}$} ($(u.0)-(0,0.5ex)$);
    \path ($(v.0)+(0,0.5ex)$) edge ($(w.180)+(0,0.5ex)$);
    \path ($(w.180)-(0,0.5ex)$) edge ($(v.0)-(0,0.5ex)$);
\end{tikzpicture}\qquad\qquad\begin{tikzpicture}[
    cir/.style={minimum size=0pt},
    box/.style={minimum size=25pt, regular polygon, regular polygon sides=4, inner sep=0pt}
]
    \node [state, box] at (  0, 0) (u) {$u$};
    \node [state, box] at (1.4, 0) (v) {$v$};
    \node [state, cir] at (2.8, 0) (w) {$w$};

    \node              at (1.4, 1) {$T_3$};

    \path ($(u.0)+(0,0.5ex)$) edge node[above]{$\phantom{C}$} ($(v.180)+(0,0.5ex)$);
    \path ($(v.180)-(0,0.5ex)$) edge node[below]{$C$} ($(u.0)-(0,0.5ex)$);
    \path ($(v.0)+(0,0.5ex)$) edge ($(w.180)+(0,0.5ex)$);
    \path ($(w.180)-(0,0.5ex)$) edge ($(v.0)-(0,0.5ex)$);
\end{tikzpicture}
\caption{Three winning strategy templates for $\eventually w$. Their permissiveness are equal but $T_1$ is larger than $T_2,T_3$.}\label{ex}
\end{figure}
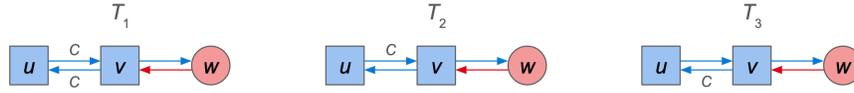
\end{example*}

Given a template $T$, the problem of constructing a smallest template $\widehat{T}$, in term of either overall or element-wise size, such that $\mathcal{L}(\psi_{\widehat{T}}) = \mathcal{L}(\psi_T)$ can be of interest when the memory of a controller is constrained, such as in a case of embedded devices. Also, small templates may yield less conflict when being composed. We \change{have} yet to explore this problem in depth and leave it as a future work.

\section{Extracting Strategies from Templates}\label{sec5}
Strategy templates constructed in Sect.~\ref{sec3}--\ref{sec99} provide useful information about restrictions on edges of the game. From these templates, our goal is to extract a strategy for player Even that satisfies the winning criterion and objective. The term \emph{template} in this section can refer to both almost-sure and positive winning strategy template.

By Def.~\ref{def:as_tem} and~\ref{def:pos_tem}, it suffices to construct a strategy $\sigma_\square$ of player Even from a winning strategy template $T$ such that $\Pr^{\sigma_\square,\sigma_\bigcirc}_{v_0}[\psi_T] = 1$ for any $v_0 \in W_\square$ (or $v_0 \in W_\square \cup W_?$) and $\sigma_\bigcirc$. The work~\cite{DBLP:conf/cav/AnandNS23} proposed the procedure \textsc{Extract}, shown below, to construct player Even's strategy $\tilde{\sigma}_\square : V^* \times V_\square \to V$ (i.e., a \emph{pure} strategy). By construction, paths generated by $\tilde{\sigma}_\square$ satisfy $\psi_T$ with probability 1.

\begin{center}\setlength{\fboxsep}{10pt}\fbox{\parbox{0.94\textwidth}{
$\textsc{Extract}(G = (V,E,(V_\square,V_\bigcirc,V_\mathlarger{\triangle})),T = (P,\LL,C))$
\begin{enumerate}[nolistsep]
    \item Remove all edges in $P$ and $C$ from $G$
    \item $\tilde{\sigma}_\square(v)$ alternates between all edges available at $v$\label{algo_extract:line2}
\end{enumerate}
}}
\end{center}

The procedure \textsc{Extract} is simple yet restrictive on the constraint of co-live edges. More precisely, a strategy template $T = (P,\LL,C)$ requires all edges in $C$ to be used only finitely often in a path. However, $\tilde{\sigma}_\square$ does not allow any usage of those edges in a path at all. Although this restriction does not affect the correctness of $\tilde{\sigma}_\square$, we prefer a winning strategy constructed to be \emph{permissive}, defined in term of formal language as follows.

\begin{definition}[Language generated by Strategy]\label{def10}
Given an SG $G$ and a strategy $\sigma_\square$ of player Even, the \emph{language generated by $\sigma_\square$} is $\mathcal{L}(\sigma_\square) \subseteq V^\omega$ containing all infinite paths $v_0 v_1 \ldots$ such that, for all $i \in \mathbb{N}$, if $v_i \in V_\square$, then $\sigma_\square(v_0 \ldots v_i)(v_{i+1}) > 0$.
\end{definition}

Definition~\ref{def10} defines generated languages for \emph{mixed} strategies (i.e., $\sigma_\square$ is a function $\sigma_\square : V^* \times V_\square \to \mathcal{D}(v)$). For pure strategies, we simply replace the last condition with ``$\sigma_\square(v_0 \ldots v_i) = v_{i+1}$''.

\begin{definition}[Permissiveness of Strategy]\label{defn:permissive}
Given an SG $G$ and two player Even's strategies $\sigma_\square$ and $\sigma'_\square$, we say that $\sigma_\square$ is \emph{no more permissive than} $\sigma'_\square$ if $\mathcal{L}(\sigma_\square) \subseteq \mathcal{L}(\sigma'_\square)$. Also, we say that $\sigma'_\square$ is \emph{more permissive than} $\sigma_\square$ if $\mathcal{L}(\sigma_\square) \subsetneq \mathcal{L}(\sigma'_\square)$.
\end{definition}

Notice that Def.~\ref{defn:permissive} does not require that a path in $\mathcal{L}(\sigma_\square)$ satisfies a winning objective. Thereby, the maximally permissive strategy is the one that allows all paths, corresponding to the template $(\emptyset, \emptyset, \emptyset)$. However, we focus only on almost-sure/positive winning strategies and aim for the strategy to be as permissive as possible. Below, we present a procedure to construct a strategy $\hat{\sigma}_\square$ from a strategy template using parameters $\alpha < 1$ and $\beta \geq 1$. These parameters balance between the permissiveness and the speed of reaching key target vertices needed to satisfy the winning objective.

Our proposed procedure \textsc{ParameterizedExtract} constructs a strategy $\hat{\sigma}_\square$ under which generated paths satisfy a strategy template with probability 1. An infinite path generated by $\hat{\sigma}_\square$ can use edges in $C$. Nevertheless, every time $\hat{\sigma}_\square$ uses an edge in $C$, the probability that it is used again becomes smaller. Hence, the probability that an edge in $C$ is used infinitely often is zero, complying with the requirement of $T = (P,\LL,C)$. We also modify the extraction procedure further by increasing the probability that an edge in live-groups in $\LL$ is used again once it is used. In this way, a path targets edges in live-groups more often.

\begin{center}\setlength{\fboxsep}{10pt}\fbox{\parbox{0.94\textwidth}{
$\textsc{ParameterizedExtract}(G = (V,E,(V_\square,V_\bigcirc,V_\mathlarger{\triangle})),T = (P,\LL,C))$
\begin{enumerate}[nolistsep]
    \item Remove all edges in $P$ from $G$
    \item For $v \in V_\square$ and $v' \in V$, set $d(v)(v') \gets 1$ if $(v,v') \in E$ and $0$ otherwise
    \item Define $\hat{\sigma}_\square(v_0 \ldots v)(v') = d(v)(v')/\sum_{v'' \in V} d(v)(v'') \in \left[0,1\right]$
    \item When $(v,v') \in C$ is used, update $d(v)(v') \gets \alpha \cdot d(v)(v')$ where $\alpha < 1$
    \item When $(v,v') \in L$ for some $L \in \LL$ is used, update $d(v)(v') \gets \beta \cdot d(v)(v')$ where $\beta \geq 1$
\end{enumerate}
}}
\end{center}

We prove the permissiveness of our winning strategy in the following theorem. We emphasize that the original procedure \textsc{Extract} of~\cite{DBLP:conf/cav/AnandNS23} considers pure strategies. However, our definition of strategies is mixed. Thus, it is not surprising that \textsc{ParameterizedExtract} can be more permissive than \textsc{Extract}. Notice also that one could generalize \textsc{Extract} to construct mixed strategies by changing Line~\ref{algo_extract:line2} to ``$\tilde{\sigma}_\square(v)$ chooses an edge available at $v$ uniformly at random''. Nonetheless, this generalization still completely prohibits the usage of co-live edges in $C$, which is allowed to be used finitely often by our proposed procedure.

\begin{theorem}
Given an SG $G$ and a strategy template $T = (P,\LL,C)$. Let $\tilde{\sigma}_\square$ and $\hat{\sigma}_\square$ follow $\textnormal{\textsc{Extract}}(G,T)$ and $\textnormal{\textsc{ParameterizedExtract}}(G,T)$, respectively. Then, $\tilde{\sigma}_\square$ is no more permissive than $\hat{\sigma}_\square$. Moreover, if there is an infinite path $\bar{v} = v_0 v_1 \ldots \in \mathcal{L}(\hat{\sigma}_\square)$ such that $(v_i,v_{i+1}) \in C$ for some $i \in \mathbb{N}$, then $\hat{\sigma}_\square$ is more permissive than $\tilde{\sigma}_\square$.
\end{theorem}
\begin{proof}
It is clear that $\mathcal{L}(\tilde{\sigma}_\square) \subseteq \mathcal{L}(\hat{\sigma}_\square)$ by construction, making $\tilde{\sigma}_\square$ no more permissive than $\hat{\sigma}_\square$. Also, assuming the existence of an infinite path $\bar{v} = v_0 v_1 \ldots \in \mathcal{L}(\hat{\sigma}_\square)$ with $(v_i,v_{i+1}) \in C$ for some $i \in \mathbb{N}$. Then, $\bar{v} \not\in \mathcal{L}(\tilde{\sigma}_\square)$ since \textsc{Extract} removes all edges in $C$ from $G$. Under this assumption, $\mathcal{L}(\tilde{\sigma}_\square) \subsetneq \mathcal{L}(\hat{\sigma}_\square)$ and thus the theorem holds.
\end{proof}

We would like to point out that equally permissive templates may induce strategies with unequal permissiveness, depending on the strategy extraction procedure. For example, let $\sigma_{\square,T_1}$ and $\sigma_{\square,T_3}$ be constructed by \textsc{Extract} using $T_1$ and $T_3$ in Fig.~\ref{ex} as inputs. Since the procedure removes all co-live edges in $C$, the path $u (v w)^\omega$ is allowed by $\sigma_{\square,T_3}$ but prohibited by $\sigma_{\square,T_1}$. As a result, although $T_1$ and $T_3$ are equally permissive, $\sigma_{\square,T_3}$ is more permissive than $\sigma_{\square,T_1}$. Note that, in contrast, \textsc{ParameterizedExtract} outputs equally permissive strategies in this specific example.

\change{We finish by discussing the practical relevance of strategy templates. While our proposed templates and algorithms remain to be implemented, we anticipate results comparable to the non-stochastic ones in~\cite{DBLP:conf/cav/AnandNS23}. In incremental synthesis---where objectives are introduced sequentially---their empirical evaluation demonstrated at least a twofold speed-up when using strategy templates to avoid recomputing strategies from scratch. For fault-tolerant control, strategy templates can effectively reduce recomputation even when up to 30\% of player Even’s choices are disabled.

In terms of computational complexity, strategy construction algorithms match the worst-case runtime of standard algorithms for solving the respective objectives. Moreover, strategy composition and extraction can be performed efficiently at runtime. These observations underscore the potential of strategy templates for broad practical applicability.}

\section{Conclusion}
This work demonstrated how strategy templates can be extended to include stochastic games. First, we introduced almost-sure winning strategy templates. We then presented several algorithms to construct such templates for stochastic game objectives. The key idea involved incorporating additional set operators of Banerjee et al.\change{~\cite{DBLP:conf/tacas/BanerjeeMMSS22}} and the gadgets of Chatterjee et al.\change{~\cite{DBLP:conf/csl/ChatterjeeJH03}} to account for player Random. Next, we defined positive winning strategy templates and proposed template constructing and composing algorithms. We later considered permissiveness and sizes of templates and discussed their relevance. Lastly, we developed a strategy extracting procedure which balances permissiveness with the speed of reaching key target vertices necessary to satisfy the winning objective.

In future work, we aim to explore the construction of smaller templates that maintain the same permissiveness as a given template. This would enhance templates' practicality across different use cases. Also, we plan to develop templates for a broader range of winning criteria and objectives, including those defined by metric temporal logic (MTL) formulae whose operations are time-constrained. This would enable the use of templates in real-time systems.





\bibliographystyle{elsarticle-num}
\bibliography{ref}







\end{document}